\documentclass[sigplan,10pt]{acmart}
\usepackage{style}
\usepackage{amsthm}
\usepackage{cleveref}
\usepackage{listings}
\usepackage{enumitem}
\usepackage{thm-restate}
\usepackage{enumitem}

\def\draft{1}  
\if\draft 0
\renewcommand{\yihan}[1]{{}}
\renewcommand{\laxman}[1]{{}}
\renewcommand{\yan}[1]{{}}
\renewcommand{\guy}[1]{{}}
\else
\fi

\newcommand{\myfunc}[1]{{\texttt{#1}}}

\newcommand{\emp}[1]{{\textbf{\textit{#1}}}}

\newcommand{\ourtree}{PaC-tree\xspace}
\newcommand{\ourname}{PaC\xspace}
\newcommand{\ctree}{C-tree\xspace}
\newcommand{\aspen}{Aspen\xspace}
\newcommand{\pam}{PAM\xspace}
\newcommand{\cpam}{CPAM\xspace}
\newcommand{\ptree}{P-tree\xspace}
\newcommand{\block}{block\xspace}
\newcommand{\B}{B}

\newcommand{\lc}{\mathit{lc}}
\newcommand{\rc}{\mathit{rc}}
\newcommand{\nil}{\textit{nil}}
\newcommand{\composite}{complex}
\newcommand{\simple}{simplex}
\newcommand{\exploded}{expanded}
\newcommand{\explode}{expand}
\newcommand{\ourgraph}{\ourtree{} (Diff)}

\newcommand{\blockednode}{flat node}
\newcommand{\compressednode}{\blockednode}
\newcommand{\regularnode}{regular node}
\newcommand{\internalnode}{\regularnode}
\newcommand{\leafnode}{leaf node}

\newcommand{\extraptr}{extra pointer}

\newcommand{\tsplit}{\myfunc{split}}
\newcommand{\join}{\myfunc{join}}
\newcommand{\union}{\myfunc{union}}
\newcommand{\intersection}{\myfunc{intersection}}
\newcommand{\difference}{\myfunc{difference}}
\newcommand{\multiinsert}{\myfunc{multi\_insert}}
\newcommand{\multidelete}{\myfunc{multi\_delete}}
\newcommand{\joinTwo}{\myfunc{join2}}
\newcommand{\expose}{\myfunc{expose}}
\newcommand{\fold}{\myfunc{fold}}
\newcommand{\unfold}{\myfunc{unfold}}
\newcommand{\refold}{\myfunc{refold}}
\newcommand{\op}{\myfunc{op}}
\newcommand{\opbase}{\myfunc{op\_base}}
\newcommand{\node}{\myfunc{node}}
\newcommand{\filter}{\myfunc{filter}}
\newcommand{\mapreduce}{\myfunc{map\_reduce}}
\newcommand{\range}{\myfunc{range}}
\newcommand{\find}{\myfunc{find}}
\newcommand{\append}{\myfunc{append}}

\newcommand{\Tl}{T_L}
\newcommand{\Tr}{T_R}
\newcommand{\true}{\mb{true}}
\newcommand{\false}{\mb{false}}

\newcommand{\decomposedtree}{decomposed tree}

\newcommand{\subsettree}{subset tree}
\newcommand{\dealswith}{processes}
\newcommand{\dealwith}{process}

\def\StartLineAt#1{\lstset{firstnumber=#1}}

\definecolor{best}{rgb}{0.0, 0.5, 0.0}
\newcommand{\best}[1]{\color{best}{\underline{#1}}}

\definecolor{mygreen}{rgb}{0.0, 0.5, 0.0}
\newcommand{\revised}[1]{{#1}}

\acmConference[PLDI'22]{ACM SIGPLAN Conference on Programming Language Design and Implementation}{June 20--24, 2022}{San Diego, California, United States}
\acmYear{2022}
\acmISBN{} 
\acmDOI{} 
\startPage{1}

\setcopyright{none}

\usepackage{booktabs}   
\usepackage{subcaption} 
\begin{document}

\title{\ourtree{}s: Supporting Parallel and Compressed Purely-Functional Collections}         



  \author{Laxman Dhulipala}
  \affiliation{\institution{University of Maryland}}
  \email{laxman@umd.edu}
  \author{Guy Blelloch}
  \affiliation{\institution{Carnegie Mellon University}}
  \email{guyb@cs.cmu.edu}
  \author{Yan Gu}
  \affiliation{\institution{UC Riverside}}
  \email{ygu@cs.ucr.edu}
  \author{Yihan Sun}
  \affiliation{\institution{UC Riverside}}
  \email{yihans@cs.ucr.edu}

\begin{abstract}
Many modern programming languages are shifting toward a functional
style for collection interfaces such as sets, maps, and sequences.
Functional interfaces offer many advantages, including being safe for
parallelism and providing simple and lightweight snapshots.
However, existing high-performance functional interfaces such as \pam{}, which are
based on balanced purely-functional trees, incur large space overheads for
large-scale data analysis due to storing every element in a separate
node in a tree.

This paper presents \ourtree{s}, a purely-functional data structure
supporting functional interfaces for sets, maps, and sequences that
provides a significant reduction in space over existing approaches.
A \ourtree{} is a balanced binary search tree which blocks the leaves
and compresses the blocks using arrays.
We provide novel techniques for compressing and uncompressing the
blocks which yield practical parallel functional algorithms for a
broad set of operations on \ourtree{s} such as union, intersection, filter, reduction, and range
queries which are both theoretically and practically efficient.

Using \ourtree{s} we designed \cpam{}, a C++ library that implements
the full functionality of \pam{}, while offering significant extra
functionality for compression.
\cpam{} consistently matches or outperforms \pam{} on a set of
microbenchmarks on sets, maps, and sequences while using about a
quarter of the space.
On applications including inverted indices, 2D range queries, and 1D
interval queries, \cpam{} is competitive with or faster than PAM,
while using 2.1--7.8x less space. 
For static and streaming graph processing, \cpam{} offers 1.6x faster
batch updates while using 1.3--2.6x less space than the
state-of-the-art graph processing system Aspen.

\end{abstract}



\maketitle

\section{Introduction}

Almost all modern programming languages include extensive support for
collections, such as sets, maps, and sequences either as libraries or
built-in data types.
Support for such collections has become the cornerstone of large-scale
data processing, as exemplified by systems such as Apache
Spark~\cite{ZahariaXWDADMRV16}.
Among the interfaces for collections, there has been a trend towards a
functional style, shying away from mutation (e.g., Spark is functional).
Functional interfaces have several advantages over mutating ones,
including being safe for parallelism, allowing safe composition,
permitting flexible implementations (e.g., using copies when helpful),
and supporting snapshots.
Supporting snapshots is particularly useful in scenarios in which a
stream of updates is being made to a collection which is concurrently being
analyzed~\cite{cheng2012kineograph,macko2015llama,Kemper0FFLMMR13,dhulipala2019low}.

Recent work~\cite{pam} has developed a purely functional library,
\pam{}, for representing sequences, ordered sets, ordered maps, and
augmented maps (defined in \cite{pam}) using balanced trees, called
\emph{P-trees}.
P-trees use path copying to perform updates, supporting functional
updates at a reasonably low cost (e.g., $O(\log n)$ per point update).
However they come at a cost of high space usage---every element
requires a node in the tree.
This is particularly problematic for large-scale data analysis, since
in large-systems memory is often the dominating cost.

\begin{figure}[!t]
  \hspace*{-0.25cm}\includegraphics[width=0.4\textwidth]{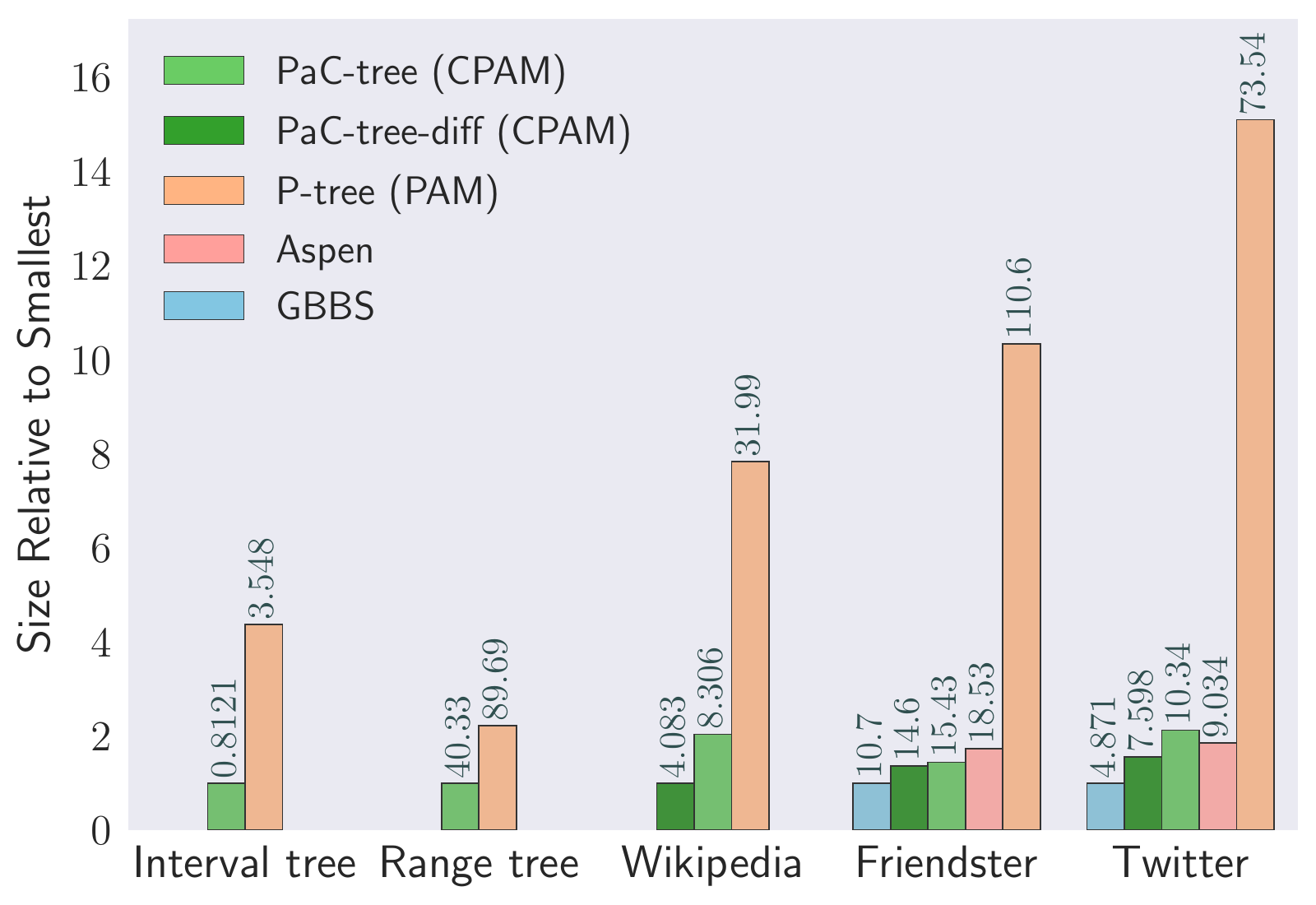}
    \caption{\small Relative sizes of the interval tree, range tree,
    inverted index (Wikipedia corpus), and graph representations
    (Twitter, Friendster) studied in this paper using \ourtree{}s from
    CPAM (using $B=128$) and other systems. Lower is better.
    The numbers shown on top of the bars are the sizes of each
    representation in GiB.
    \ourtree{}-diff compresses integer keys using difference encoding.
    The \ctree{}s from Aspen~\cite{dhulipala2019low} also support
    difference encoding.  GBBS is the \emph{static} compressed graph
    representation from the Graph Based Benchmark
    Suite~\cite{DhulipalaBS21} which uses difference
    encoding, and serves as a baseline for the tree-based graph
    representations.
      }
    \label{fig:size-comparison-intro} \vspace{-1em}
\end{figure}

In this paper we present \defn{Parallel Compressed trees
(\ourtree{}s)}: a purely-functional data structure for supporting a
similar functionality as P-trees but with significant reduction in
space---up to an order of magnitude (see
\cref{fig:size-comparison-intro}).  Our approach is based on
blocking the leaves and compressing the blocks using arrays (see \cref{fig:ourtree}).
We present innovative techniques for compressing and uncompressing the
blocks without needing to re-implement the full functionality of
P-trees.
Importantly, in the paper we analyze the cost of all the operations as
a function of the block size $B$ as well as the collection size.  
This is analyzed both in terms of the work (runtime
sequentially) and span (longest dependent path in parallel).  The
costs for a sample of the supported functions are given in
\cref{table:algorithmcosts}.  These costs can help the user
decide on a block size for their particular application---a parameter
that can be specified when creating a collection.

\begin{table}[tbp]
  \footnotesize
  \centering

\hspace*{-0.5em}
\begin{tabular}[!t]{llcc}
    \toprule
    & \textbf{Primitive} & Work & Span \\
    \midrule
    \parbox[t]{2.5mm}{\multirow{9}{*}{\rotatebox[origin=c]{90}{\bf Sequence}}}
    & \emph{Build}  &  $O(n)$  &  $O(\log n)$    \\
    & \emph{Map}  &  $O(n)$ &  $O(\log n)$  \\
    & \emph{Filter}  &  $O(n)$ &  $O(\log n)$ \\
    & \emph{Reduce}  &  $O(n)$ &  $O(\log n)$ \\
    & \emph{Take}    &  $O(\log n + B)$  &  $O(\log n)$ \\
    & \emph{$n$-th}  &  $O(\log n + B)$  &  $O(\log n)$ \\
    & \emph{FindFirst}  &  $O(k)$ &  $O(\log n)$        \\
    & \emph{Append}$^{\dagger}$   & $O(\log n + B)$  & $O(\log n)$  \\
    & \emph{Reverse}$^{\dagger}$  & $O(n)$  &  $O(\log n)$   \\
    \midrule
    \parbox[t]{2.5mm}{\multirow{8}{*}{\rotatebox[origin=c]{90}{\bf Set and Map}}}
    & \emph{Build}  &  $O(n \log n)$  &  $O(\log n)$   \\
    & \emph{Next/Previous}  &  $O(\log n + B)$ &  $O(\log n)$  \\
    & \emph{Rank}  &  $O(\log n + B)$  &  $O(\log n)$  \\
    & \emph{Range}  &  $O(\log n + B)$ &  $O(\log n)$\\
    & \emph{Insert}  &  $O(\log n + B)$ &  $O(\log n)$ \\
    & \emph{Union}  &  $O(m\log\frac{n}{m} + \min(mB, n))$ &  $O(\log n \log m)$ \\
    & \emph{Intersect}  &  $O(m\log\frac{n}{m} + \min(mB, n))$ &  $O(\log n \log m)$  \\
    & \emph{Difference}  &  $O(m\log\frac{n}{m} + \min(mB, n))$ &  $O(\log n \log m)$  \\
    \bottomrule
  \end{tabular}
  \par
    \caption{\small \textbf{Primitives from the Sequence, Set, and Map
    interfaces in CPAM, including the work and span bounds.}
    Note that primitives marked with $^{\dagger}$ are specific to
    Sequences, and Set and Map primitives cannot be applied to
    Sequences.
    $m,n$ are defined to be the size of the smaller and larger sets,
    respectively. $B$ is the block size (the size of a \block{ed} leaf
    in a \ourtree{}). We assume a parallelizable encoding for the
    span bounds.
    }
    \label{table:algorithmcosts}
\end{table}

Using \ourtree{}s we have implemented \cpam{}: a C++ library which
implements the full functionality of \pam{}, along with significant
extra functionality involving compression.
By default \cpam{} supports difference (or delta)
encoding~\cite{Compression} within the blocked leaves.  In such an
encoding, each element is encoded based on the value of the previous
element in the collection.
This can greatly reduce space when elements that are close in the
ordering of the collection are related.  For example, if a graph is
numbered so that neighboring vertices have similar indices, then the
neighbors in a neighbor list will have small differences.
These small numbers can then be encoded in a handful of bits
each~\cite{shun2015smaller}.
Similarly in an inverted index where each word points to a sequence of
documents it appears in, if the documents are sorted, the differences
between adjacent document identifiers can be small.
This is especially true for common words, which take up the bulk of
the space.  In the paper we bound the extra space needed (due to the
index using the tree structure) for \ourtree{s} compared to a static
representation of the data (i.e., an array) directly using difference
encoding (see \cref{thm:space}).

In our default blocked representation, the first element of a block is
represented uncompressed, and the rest of the elements are compressed
relative to the previous element.
In addition to delta-encoding, \cpam{} also supplies an interface for
the user to define their own form of compression for each block.
For example, they can quantize values, or use other variable length
codes when keys are known to be small.
%
\cpam{} uses a reference counting garbage collector to manage the
memory for both the internal nodes and the compressed leaf nodes,
which can be of variable size due to compression.

\cpam{} supports augmentation in which each tree node maintains
an aggregate of the values of its subtree (see more details in
\cref{sec:prelim}).
The aggregation function is declared as part of the type of the tree.
Augmentation is useful in many applications, and indeed we use it in
all of the applications we describe later.
\ourtree{}s store an augmented value per internal node, and one for
each block at the leaves. Storing one value per block significantly
reduces space relative to \ptree{}s in \pam, which store a value for
every element.

\begin{figure}[!t]
  \hspace*{-0.25cm}\includegraphics[width=0.5\textwidth]{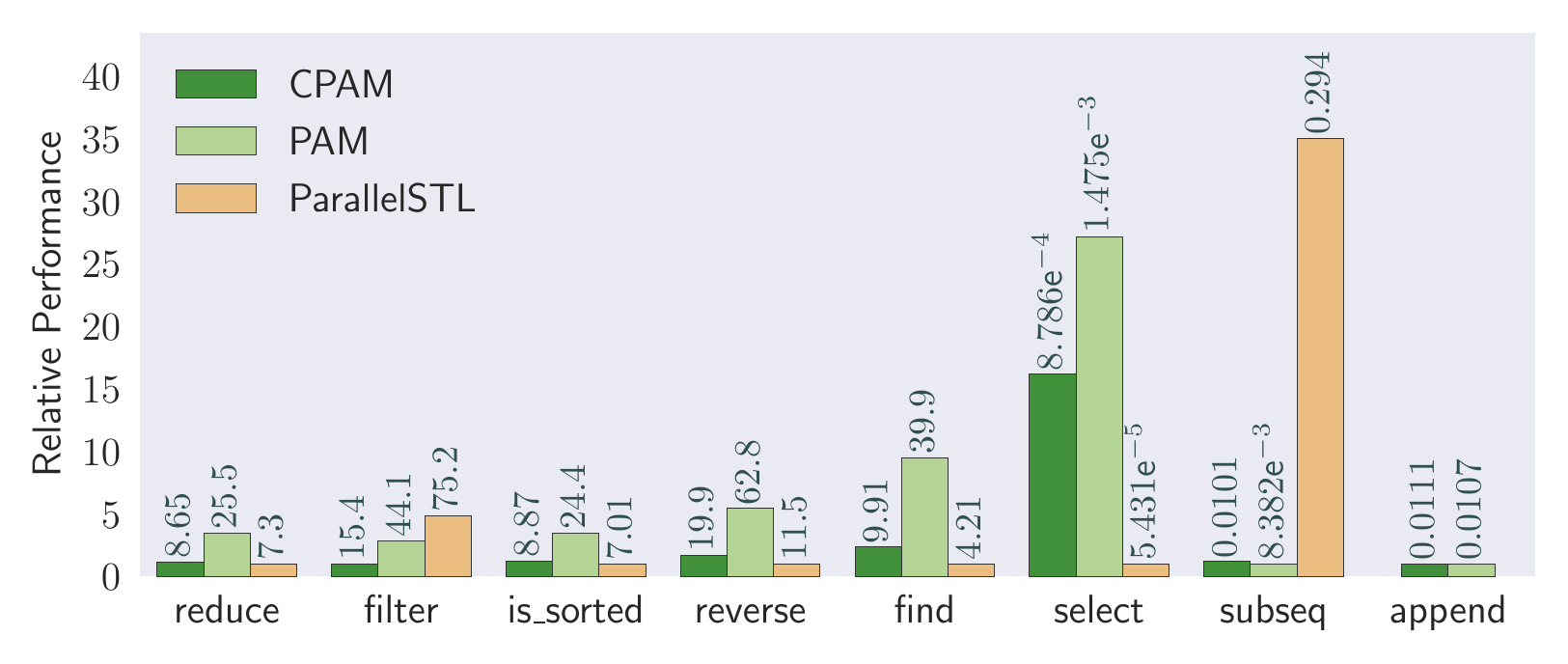}
    \caption{\small Relative performance of sequence primitives in
    CPAM (using $B=128$), PAM, and
    ParallelSTL~\cite{ISOCPPparallelism} on a 72-core machine with
    2-way hyper-threading enabled. The numbers shown on top of the
    bars are the parallel (144-thread) running times in milliseconds. Lower is
    better. All benchmarks are run on sequences of length $10^{8}$
    containing 8-byte elements. For \append{},
    ParallelSTL takes 17.7 milliseconds on average (1594x larger than
    \append{} in CPAM). CPAM and PAM represent sequences using
    purely-functional trees, whereas ParallelSTL uses arrays (hence
    static).
  }
\label{fig:comp}
  \vspace{-1em}
\end{figure}

To demonstrate the effectiveness of \ourtree{}s, and their
implementation in \cpam{}, we measure performance and space usage on
(1) a collection of microbenchmarks that directly use some of the
functions supported by the library, and (2) a handful of real-world
applications.

For the microbenchmarks, we compare the performance of \cpam{} to
\pam{}, and for sequences to the Intel implementation of the C++17
parallel STL library~\cite{ISOCPPparallelism} (ParallelSTL).
ParallelSTL is a highly optimized library supporting only sequences
based on arrays.
A summary of the results for sequences is given \cref{fig:comp}, and
details including performance of ordered maps, and augmented maps are
given in \cref{sec:exps}.
Compared to \pam{}, \cpam{} achieves significantly better performance
due to the reduced memory footprint, and hence reduced number of cache
misses, while only requiring about 1/4-th as much space even without
compression.
Compared to ParallelSTL, \cpam{} has similar performance on operations that
visit the whole sequence, like \texttt{reduce}, but is significantly
slower on \texttt{nth} since it requires $O(\log n + B)$ work as
opposed to $O(1)$ for a random array access for ParallelSTL.
On \texttt{append} \cpam{} is significantly faster since it requires
$O(\log n +B)$ work to join to trees instead of $O(n)$ required by
ParallelSTL to copy the input arrays into the output array.

We consider four applications: graphs, inverted indices, 2D range
queries and 1D interval queries.
For inverted indices, 2D range query and 1D interval query, \cpam{}
achieves competitive performance to \pam{} while using 2.1x--7.8x less
space.
For graph processing, we compare to an existing system
\aspen{}~\cite{dhulipala2019low} that represents graphs using trees.
\cpam{} uses 1.3--2.6x less space compared to \aspen{}, and is almost
always faster than \aspen{} in all tested graph algorithms.

The main contributions of this paper are:
\setlength{\itemsep}{0pt}
\begin{itemize}[topsep=1.5pt, partopsep=0pt,leftmargin=*]
  \item A new functional data structure, \ourtree{}s, and associated
    parallel algorithms that support compression for sequences, sets,
    maps and augmented maps.
  \item Theoretical bounds on the costs (work and span)
    and the space of the data structure and associated algorithms.
  \item An implementation of \ourtree{}s as a library, \cpam,
   supporting the full functionality of \pam{} in addition to
    supporting default and user defined compression
    schemes.\footnote{We have made CPAM publicly
    available: \url{https://github.com/ParAlg/CPAM}.}
 \item An experimental evaluation of the ideas and implementation on
  microbenchmarks and non-trivial applications.
\end{itemize}

\section{Related Work}
\label{sec:related}

Our work extends \ptree{}s and their C++ implementation in \pam{}~\cite{pam}.   Our
key contribution is the ability to compress the trees achieving up to
an order-of-magnitude reduction in space.  This is achieved while
being able to present cost bounds both in terms of time and space.
These bounds are a function of a block size the user can select.

B-trees~\cite{Bayer1972} and their variants block not just the leaves
but all nodes of a tree, such that internal nodes can have a high
fan-out.  They are widely used in practice, especially for disk based
data structures since nodes are on the scale of a page on disk and can
be retrieved efficiently.
However they are less relevant in the context of purely functional
in-memory trees.
In particular, path copying requires that an update copy all nodes on
the path from the root to the leaf.  If the nodes are large (e.g. 128+
elements each, as in our leaves) this copying would be very expensive
both in terms of space and time.
\revised{Various work has suggested blocking the leaves of a binary tree to
represent sequences~\cite{Acar14,Fluet08,Boehm95,yi08,kmett10}.
The idea is to reduce the cost of operations such as append or
subsequence relative to array representations.  As far as we know,
these ideas have never been applied to ordered sets or ordered maps.\footnote{We note that the design of the chunked sequence datatype~\cite{Acar14}
could in principle be extended to support sets, maps, and augmented
maps, although the implementation is specialized for ephemeral
sequences.}
We also do not know of work that then compresses within the blocks.
}

\aspen{}~\cite{dhulipala2019low} is a system for graph processing,
based on purely functional trees and uses compression for the neighbor
lists.  At a high-level, our goals are shared with \aspen{} (e.g.,
non-mutating updates), but \aspen{} has several limitations.
Importantly it is only designed for graphs, supporting only a small
part of the functionality of \cpam{}.
The tree representation in \aspen{} is also very different.  It
randomly selects elements from the collection to be \emph{heads}.  It
then attaches a block of nodes to each head corresponding to the keys
between the head and the next head, and puts the heads into a binary
tree.
\revised{\ourtree{}s do not require randomization, and have stronger
theoretical bounds for primitive operations such as \union{} than the
bounds provided by \ctree{}s in \aspen{}.  We use \cpam{} to implement
the full functionality of \aspen{} and compare to \aspen{} in
\cref{sec:graphs}.}

\begin{figure}
  \centering
  \includegraphics[width=\columnwidth]{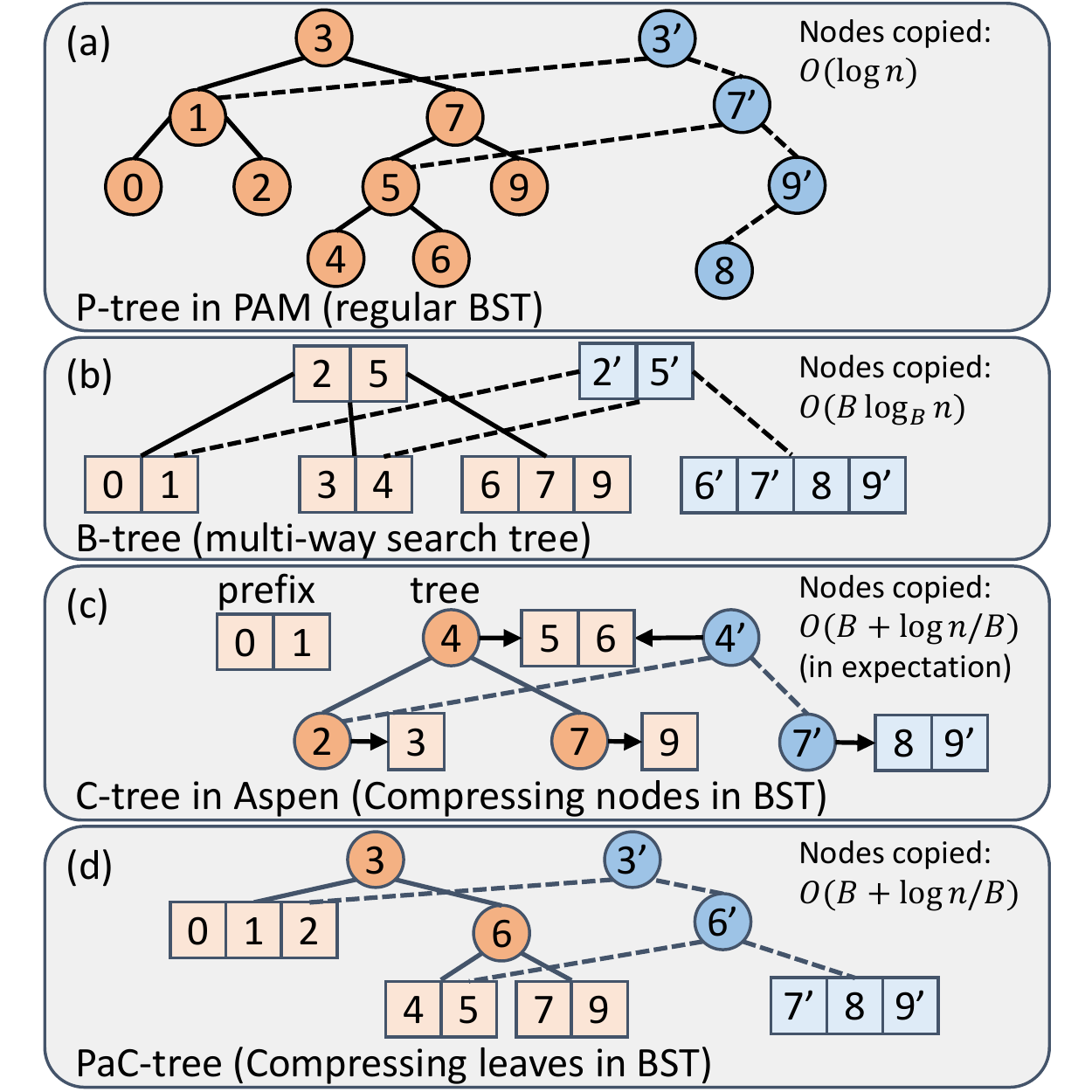}
  \caption{\small {An illustration of (a) \ptree{} in \pam \cite{blelloch2016just,pam} (regular BST), (b) B-tree (multi-way search tree), (c) \ctree \cite{dhulipala2019low} in \aspen{} (compressing all nodes in a BST) and (d) our \ourtree{} (compressing all leaves in a BST) in \cpam{}.} The orange nodes show a tree with keys 0-7 and 9. We then consider inserting a key 8. Blue nodes are what we need to create (copy or new) due to path-copying. Round nodes are tree nodes each storing a single key, and square nodes are organized in blocks of size $O(B)$ (expected for C-trees).
  Let $n$ be the tree size, an insertion needs to copy $O(\log n)$ nodes in \ptree{}, $O(B \log_Bn)$ in B-tree, and $O(B+\log(n/B))$ in \ctree{} (in expectation) or \ourtree{}.
  }\label{fig:alltrees}
\end{figure}

\cref{fig:alltrees} compares \ptree{}s from \pam, functional B-trees,
C-trees from \aspen, and \ourtree{}s.  The comparison illustrates how
they differ when inserting a new key.

Like \cpam{}, the Apache Spark~\cite{ZahariaXWDADMRV16} system
supports a functional interface for collections.   However it has
significant differences.   Firstly it only supports unordered sets.
Secondly although it has a shared-memory parallel implementation, it
is primarily designed for a distributed setting.  This means its
shared-memory implementation is not ideal.\footnote{Their shared-memory
implementation is between 3.2--4.9x slower than \cpam{} for a
map, reduce, and group-by style example taken from their user guide.
For primitives such as map and reduce, their implementation performs
up to 2 orders of magnitude worse than \cpam{} 
\ifx\confversion\undefined
(see \cref{sec:spark})}
\else
(see the full version).
\fi

\revised{There is extensive research on concurrent tree data
structures~\cite{ellen2010non,braginsky2012lock,Brown18,KungL80,Natarajan14,bronson10,aksenov2017concurrency}.
This work is mostly orthogonal to our work.
Such trees support a fraction of the functionality of \cpam{},
typically just supporting linearizable inserts, deletes, updates and
finds.   Some recent work support range
queries~\cite{fatourou2019persistent,basin2017kiwi}, or arbitrary
queries on a snapshot~\cite{WeiBBFR021}.  On the other hand concurrent
trees support asynchronous updates, which \ourtree{}s do not---such
updates are inherently non-functional.    To support multiple
concurrent updates, \ourtree{}s would require batching the update and
applying as a batch in parallel (fairly comparing concurrent and
batched structures like \ourtree{} seems challenging for this reason).
We expect the use cases would be quite different.}

Blandford and Blelloch developed tree structures for ordered sets that support compression~\cite{BB04}.  They present space bounds that are similar to ours, in terms of relating the space of a difference encoded sequence to the space of the data structure.  However they support a small fraction of the functionality described in our work.

Functional trees using path-copying date back to at least the early 1990s~\cite{adams1993functional}, and in the sequential
setting have been studied by Kaplan and Tarjan~\cite{KaplanT96} and Okasaki~\cite{okasaki1999purely}. 
\section{Preliminaries}
\label{sec:prelim}
\myparagraph{Binary search trees.} A \emph{binary search tree} (BST) is either an empty node, denoted as \nil{}, or a
node consisting of a \emph{left} BST $T_L$, a key $k$ (or with an associated value), and
a \emph{right} BST $T_R$, denoted \node$(\Tl,k,\Tr)$, where $k$ is larger than all keys in $T_L$
and smaller than all keys in $T_R$.
We use $\lc(T)$ and $\rc(T)$ to extract the left and right subtrees of $T$, respectively, and use $k(T)$
to denote the key stored at $T$'s root.
The \emph{size} of a BST $T$, or $|T|$, is the number of nodes in $T$.
The \emph{weight} of a BST $T$, or $w(T)$, is $1+|T|$.
The \emph{height} of a BST $T$, or $h(T)$,
is $0$ for \nil{}, and $\max(h(\lc(T)), h(\rc(T))) + 1$ otherwise.
\hide{\emph{Parent}, \emph{child}, \emph{ancestor} and
\emph{descendant} are defined as usual (ancestor and descendant are
inclusive).}
A tree node is a leaf if it has no children, and a \internalnode{} otherwise.
The \emph{left (right) spine} of a binary tree is the path of nodes from the root to a \nil{} node, always
following the left (right) tree.
\hide{
A \emph{binary tree} is either an empty node, denoted as \nil{}, or a
node consisting of a \emph{left} binary tree $T_L$, an entry $e$, and
a \emph{right} binary tree $T_R$, denoted \node$(\Tl,e,\Tr)$.
The \emph{size} of a binary tree, or $|T|$, is the number of nodes in $T$.
The \emph{weight} of a binary tree $T$, denoted as $w(T)$, is $1+|T|$.
The \emph{height} of a binary tree, or $h(T)$,
is $0$ for a \nil{} node, and $\max(h(\Tl), h(\Tr)) + 1$ for a
\node$(\Tl,v,\Tr)$.
\emph{Parent}, \emph{child}, \emph{ancestor} and
\emph{descendant} are defined as usual (ancestor and descendant are
inclusive).
The \emph{left (right) spine} of a binary tree is the path of nodes from the root to a \nil{} node, always
following the left (right) tree.
The \emph{in-order values} of a binary tree
is the sequence of values returned by an in-order traversal of the
tree.
A tree is a \emph{binary search tree} when the entries $e$ contains a key, noted as $k(e)$, and the in-order values of the tree is a sorted sequence with respect to a total ordering on keys. In this work, we always assume the binary trees are binary search trees maintaining ordered sets or maps.
}

A \emph{weight-balanced tree}, or BB$[\alpha]$ trees~\cite{nievergelt1973binary} is a BST where for every
$T=\node(\Tl,v,\Tr)$, $\alpha \le \frac{w(\Tl)}{w(T)}\le 1-\alpha$.
We omit the parameter $\alpha$ with clear context.
\hide{We say two
weight-balanced trees $T_1$ and $T_2$ have \emph{balanced} weights if
$\node(T_1,e,T_2)$ is weight balanced.  }
A weight-balanced
tree $T$ has height at most $\log_{\frac{1}{1-\alpha}}w(T)$.
\hide{
For
$\frac{2}{11} < \alpha \le 1-\frac{1}{\sqrt{2}}$ insertion and
deletion can be implemented on weight balanced trees using just single
and double rotations~\cite{weightbalanced,blum1980average}.  We
require the same condition for our implementation of \join{}, and in
particular use $\alpha=0.29$ in experiments.}

\myparagraph{Parallelism.}
Our implementation of \ourtree{s} is based on nested fork-join parallelism~\cite{CLRS,frigo1998implementation,Java-fork-join}.
We analyze our algorithms use work-span model based on binary-forking~\cite{BlellochF0020}.
The \emph{work} $W$
of a parallel algorithm is the total number of operations, while
the \emph{span} is the critical path length of
its computational DAG.
We use $s_1~||~s_2$ to indicate
that statements $s_1$ and $s_2$ can run in parallel.
Almost all algorithms use divide-and-conquer to enable parallelism.
Any computation
with $W$ work and $S$ span will run in time $T < \frac{W}{P} + S$ on
$P$ processors assuming shared memory and a greedy scheduler
\cite{Brent74,blumofe1999scheduling}.
We use $\log n$ to denote $\log_2(n+1)$ in the cost bounds.

\myparagraph{Encoding schemes.}
We use \defn{Difference Encoding} (DE) to encode integer keys.
Given a sorted set of keys, $K$, the
difference encoding scheme stores the differences between consecutive
keys using an integer code, such as byte or $\gamma$ codes. We only
consider byte codes in this paper since they are cheap to encode and
decode and do not waste much space compared to using $\gamma$
codes~\cite{shun2015smaller}.


\hide{
\section{Background and Technical Overview}
\label{sec:overview}

Here we discuss some previous work and an overview of \ourtree{s}.
We note that there are a rich literature on parallel and concurrent tree
structures \cite{agrawal2018parallel,basin2017kiwi,hassan2015transactional,tian2019transforming,aksenov2017concurrency,brown2020non}, but they did not consider compression
or SI with batch updates.
There are also many parallel graph processing systems~\cite{gilbert2021performance,kabir2017shared,ediger2012stinger, feng2015distinger,
green2016custinger,Yin2018,cheng2012kineograph,macko2015llama},
but they did not consider supporting a system allowing for general ordered map interface.
Therefore, we focus on the two most relevant previous work,
which are \ptree{s}
in \pam{} and \ctree{s} in \aspen{}.
In 2016, Blelloch et al.
proposed \emp{\ptree{s}} for implementing parallel ordered
maps, which is later
implemented in a library PAM~\cite{pam,sun2019supporting}.
In 2019, Dhulipala et al. \cite{dhulipala2019low} extended
\ptree{s} to allow for compression (see details below), called \emp{\ctree{s}},
which are implemented in a graph streaming system \aspen{}.


\hide{
\myparagraph{\ptree{s} and \ctree{s}.}
\hide{In 2016, Blelloch et al.
proposed an algorithmic framework for implementing parallel ordered
maps using balanced binary trees.
Such data structures are later referred to as \ptree{s} and are
implemented in a library PAM~\cite{pam,sun2019supporting}.
}
In 2016, Blelloch et al.
proposed \ptree{} for implementing parallel ordered
maps, which is later
implemented in a library PAM~\cite{pam,sun2019supporting}.
\hide{\ptree{} supports algorithms including insertion and deletion,
union/insersection/difference on keys in two trees, other aggregate
functions like \filter{}, \range{}, and \mapreduce{}, as well as
supporting augmentation.
}
\hide{These algorithms are based on two primitives
\join{} and \expose{}, which extend to multiple balancing schemes
including weight balanced trees.  \ptree{s} are purely functional, and
support multi-versioning and SI using path-copying (see details
below).}  In 2019, Dhulipala et al. \cite{dhulipala2019low} extended
the data structure to allow for compression (see details below), called \emp{\ctree{s}},
which is implemented in a graph streaming system \aspen{}.

\hide{
To extend a regular binary
tree to a \ctree{}, one selects a random subset of entries from all
the entries, which are called \emph{heads}. All entries between two
heads $k_1$ and $k_2$ are maintained in an array (possibly encoded)
associated with $k_1$, called \emph{chunks}.  A \ctree{} is then
represented as a \emph{prefix}, which is the chunk before the first
head, and a balanced binary tree, where each node contains a head and
its consecutive chunk (See \cref{fig:alltrees} (c)).  This data
structure has been implemented in a library \aspen{}, and has been
shown to result in a state-of-the-art graph streaming system.
}
}
}

\myparagraph{Functional data structures.}
\ourtree{s} are purely functional data structures.
In functional data structures values are immutable, so updates
must be made by copying parts of the structure.  For search
trees, only the path to the update location needs to be copied.  Hence
for balanced trees of size $n$, single point updates such as inserts
and deletes involve copying $O(\log n)$ nodes (\cref{fig:alltrees}(a)).
This also applies to multi-point updates.
For example, if a \filter{} ends up removing a single element, only
$O(\log n)$ nodes need to be copied.  
Functional trees can also easily support multiversioning with low time
and space overhead~\cite{Ben-DavidBFRSW21,sun2019supporting}. Because
the data are immutable, any operation accesses the tree in an isolated
version.  Updates can be applied in \emph{batches} in parallel and yield a
new version. This enables all read-only queries to be performed at the
same time without being affected by ongoing (concurrent) updates.  In
addition to multiversioning, functional data structures also allow for
multiple histories.

\myparagraph{Join-based algorithms.}
\ourtree{s} are implemented using the \emph{join-based} approach~\cite{blelloch2016just,pam,sun2019parallel,sun2019supporting,GSB18,BGSS18} first
implemented in \pam{}~\cite{pam}.
In the framework, a variety of tree algorithms are implemented based
on two primitives, \join{} and \expose{}.\footnote{\pam{} did not
explicitly use \expose{} as a primitive, but only conceptually
  treated it as a primitive.}
Given a balancing scheme $\mathcal{S}$, the \join$(\Tl,e,\Tr)$
function returns a balanced tree $T$ satisfying $\mathcal{S}$ which
has the same in-order values as \node$(\Tl,e,\Tr)$.
In other words, it concatenates $\Tl$ and $\Tr$ by an entry $e$ in the
middle while preserving the balancing invariants (see
    \cref{fig:primitive}
    as an example of joining two \ourtree{s}).
The \expose{}$(\Tl)$ function returns a triple $(\Tl,e,\Tr)$, where
$e\in T$ is an entry, $\Tl$ and $\Tr$ are two binary trees such that
both $\Tl$ and $\Tr$ satisfy $\mathcal{S}$, are balanced with each
other under $\mathcal{S}$, and $\Tl$ ($\Tr$) contains all keys in $T$
that go before (after) $e$ in $T$'s in-order value.
It has been shown that on weight-balance trees with $\alpha\le
1- 1/\sqrt{2}$, a \join{} operation can be done in $O\left(\log
\frac{n}{m}\right)$ work~\cite{blelloch2016just}, where
$n=\max(|\Tl|,|\Tr|)$ and $m=\min(|\Tl|,|\Tr|)$.

Based on \join{} and \expose{}, many parallel tree algorithms can be
expressed in a simple and elegant recursive style 
\ifx\confversion\undefined
(see \cref{fig:mainalgo}, \cref{fig:otheralgotext}, and
\cref{fig:setcode} for examples).
\else
(see \cref{fig:mainalgo} and \cref{fig:otheralgotext} for examples).
\fi
We adopt the \join-based approach in our implementation, and in
particular carefully designed \join{} and \expose{} functions for
\ourtree{s}. This greatly simplifies the implementation and
correctness arguments of our algorithms. We give more details in \cref{sec:algo,sec:theory}.

\hide{
\emp{Multi-version concurrency control (MVCC) using pure functional trees.} Purely-functional (mutation-free) data structures preserve previous versions of themselves when modified, and yield a new version of the data structure reflecting the update \cite{okasaki1999purely}. Both P-trees and C-trees are purely-functional using \emph{path-copying}. The advantage of making trees purely-functional is that, one can easily support multi-version concurrency control with low time and space overhead. Because the tree is purely-functional, any operation accesses the tree in an isolated version.
Updates are applied in \emph{batches} in parallel and yield a new version. All read-only queries can be performed at the same time without being affected by ongoing updates.
}

\myparagraph{Augmentation.} An \emph{augmented tree} is a search tree where each node maintains an aggregated value (called \defn{augmented values}) of all entries in its subtree.
Typical examples would be a weighted sum, minimum or maximum of values, where we can obtain the augmented value in a node by combining augmented values of the children and itself.
This generalizes to all associative operations.
\ourtree{s} support generic user-defined augmentation for any associative operations. An example of \ourtree{} with augmentation is shown in \cref{fig:ourtree}.


\hide{
A lot of the complication comes in making this representation work
with PAM.  PAM supports several dozen operations on sequences, sets,
maps and augmented maps, and it would be significant work to write
special case implementations for all of them.  Instead, we only re-implemented
the \join{} and \expose{} functions, and all the high-level algorithms remain the same.
This leaves the implementations
mostly unchanged.
We found that the overhead is not large, but for
the most important and frequently-used operations,
we implemented special base cases for the blocked
leaves.  Some of theoretical results also require
special base cases.
}

\section{PaC-Trees}\label{sec:tree}

In this paper, we propose \ourtree{s} to support purely functional collections, which
support \emph{parallelism, determinism, compression, augmentation, strong theoretical bounds, and multi-versioning}.
\ourtree{s} are purely functional.
The base
data structure of a \ourtree{} is a weight-balanced BST.
The internal nodes remain binary so they are
cheap to copy.
The leaves in a \ourtree{} are organized in \block{s} of size $B$ to $2B$ for some parameter $B$.
An illustration is shown in \cref{fig:alltrees}.
If the \block{s} grow too large, they
are split, and if they become too small they are merged with a
neighboring node.

\begin{definition}[\ourtree{}]
A \ourtree{} $\mathcal{\mb{\ourname}}(\alpha, B,\mathcal{C})$,
parameterized by the balancing factor $\alpha$, \block{} size $B$,
and encoding scheme $\mathcal{C}$ satisfies the following invariants:

\setlength{\itemsep}{0pt}
\begin{itemize}[topsep=1.5pt, partopsep=0pt,leftmargin=*]
  \item (\textbf{Weight Balance}) For any tree node $v\in T$, $\alpha \le
  \frac{w(v_{*})}{w(v)}\le 1-\alpha$, where $\alpha\le
  1-\frac{1}{\sqrt{2}}$ is a constant, and $v_{*}$ is either $\lc(v)$
  or $\rc(v)$. Unless mentioned otherwise, we use $\alpha=0.29$.

  \item (\textbf{Blocked Leaves})\hide{If $|T|\ge B$, each \leafnode{} $u\in T$ maintains $B$ to $2B$
  entries using encoding scheme $\mathcal{C}$ in an array, as well as
  the size of the array as meta-data. When $\mathcal{C}$ is empty, it means
  the entries are directly stored in the array without compression.}
  If $|T|\ge B$, each leaf $u\in T$ maintains $B$ to $2B$
  entries in an array (called a \defn{\block{}}) using the encoding scheme $\mathcal{C}$.
  Unless mentioned otherwise, we assume $\mathcal{C}$ is empty, which means the entries are blocked without
additional compression of the entries.
\end{itemize}
\end{definition}

When the context is clear, we omit $\alpha$, $B$ and $\mathcal{C}$ in the
definition and simply call it a \emp{\ourtree}.
We call a leaf node containing multiple entries in a \ourtree{} a
\defn{\blockednode{}}, and a node
containing a single entry a \defn{\regularnode{}}.
We say a \ourtree{} (or a subtree) $T$ is a \defn{\simple{}} tree if
$|T|<B$, and thus $T$ only contains \regularnode{s}. We say a
\ourtree{} (or a subtree) $T$  is a \defn{\composite{}} tree if $T$ contains both
\regularnode{s} and \blockednode{s}.
We define the \emp{\exploded} version of a \ourtree{} $T$ (or a \blockednode{} $v$) to be a
regular binary tree (without \blockednode{s}), where
all \blockednode{s} in $T$ (or $v$ itself) are fully expanded as perfectly-balanced
binary trees. In \cref{fig:ourtree}, we show an example of an expanded tree.

\begin{figure}
  \centering
  \includegraphics[width=\columnwidth]{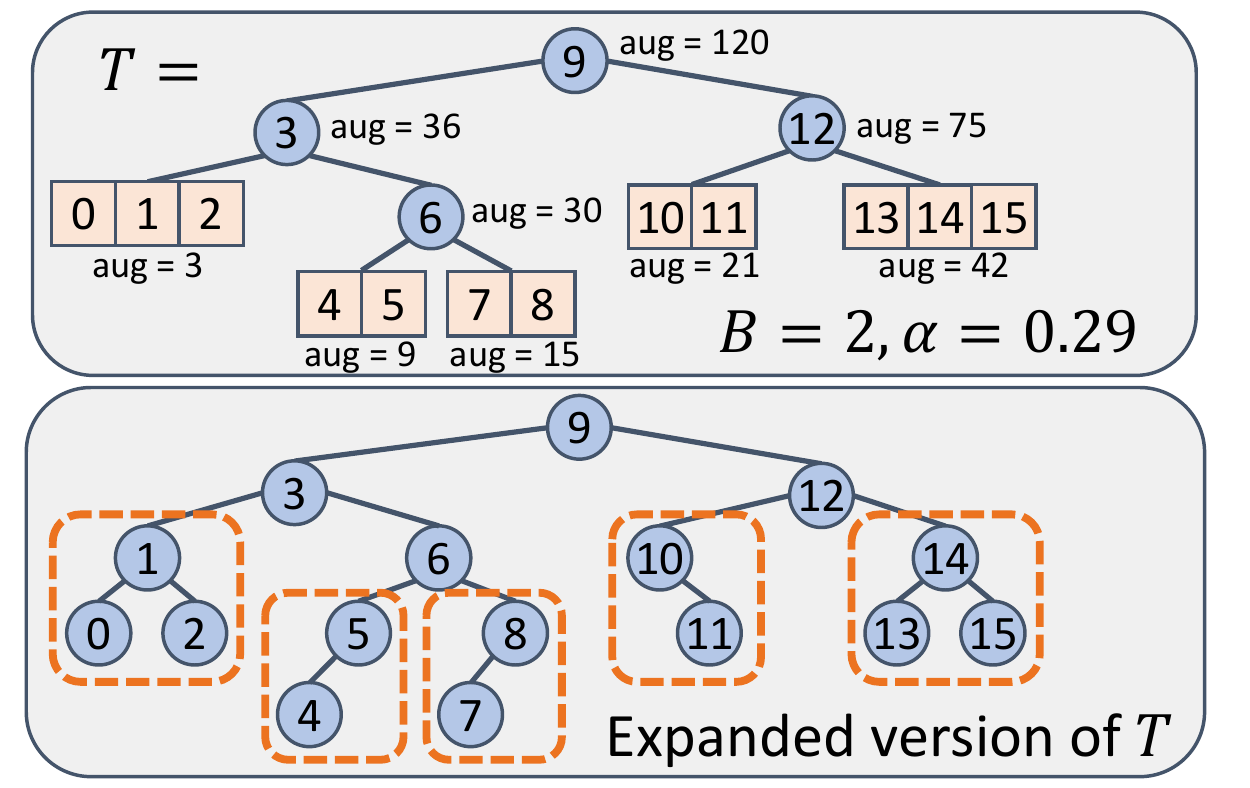}
  \caption{\small (a). An illustration of a \ourtree{} with keys $\{0,1,\dots,15\}$, and augmentation as sum of keys.
  All nodes are weight-balanced. All leaves are \block{ed} as arrays of size $B$ to $2B$.
   (b) The \exploded{} version of the \ourtree{} in (a).
   }\label{fig:ourtree}
\end{figure}

\hide{In our proofs and implementation, we assume all \blockednode{s}
maintains a sorted order of all entries. In this case, one can
build a perfectly balanced tree from a \blockednode{}, or flatten a \ourtree{} into a \block{}
in $O(B)$ work and $O(\log B)$ span. Later in the appendix, we relax \ourtree{s} to support
unsorted \blockednode{s}, which is useful for certain special
queries.
}

\hide{
In \ourtree{}, the definitions of \emph{weight}, \emph{size}, etc. are
defined similarly as on a regular weight-balanced trees (as in
\cref{sec:prelim}).
}

\hide{
We first present some useful observations about \ourtree{}.

\begin{observation}
\label{thm:leafcompress}
  For a \ourtree{} $T$ where $|T|\ge \B$, all leaves are \compressednode{s}.
\end{observation}

\begin{observation}
  In a \composite{} \ourtree{}, the path length from the root to any \leafnode{} is $O(\log n/B)$.
\end{observation}

\begin{observation}
  All \internalnode{s} in a \ourtree{} form a weight-balanced tree.
\end{observation}

}

We now present the space bound of a \ourtree{}. For integer keys, we can use difference encoding to bound the space.

\begin{theorem}
\label{thm:space}
  The total space of a \ourtree{} $\mathcal{\mb{\ourname}}(\alpha, B,\mathcal{C}_{DE})$ maintaining a set $E$ of integer keys is $s(E)+O\left(|E|/B+B\right)$, where $\mathcal{C}_{DE}$ is difference encoding, and $s(E)$ is the size needed for $E$ using difference encoding.
\end{theorem}
\begin{proof}
  The space needed for a \ourtree{} includes the \internalnode{s} and the \leafnode{s}. First of all, when $|E|<B$, all entries are maintained in a \simple{} tree, taking $O(B)$ space. When $|E|\ge B$, there are $O(|E|/B)$ \internalnode{s}, each taking $O(1)$ space for meta-data (pointers, size, etc.). The total space used by \internalnode{s} is $O(|E|/B)$. All the \leafnode{s} are organized in blocks. Let $A$ be an array that stores all keys in $E$ using difference encoding.
  Comparing the total size of all the blocks and $A$, the only extra space is the first element of each block (which cannot be compressed). There are $O(|E|/B)$ such blocks, and thus the extra space used is $O(|E|/B)$.
\end{proof}

We note that this bound is deterministic, as opposed to the bound for
\ctree{s} (which only holds in expectation). Furthermore, using known facts
about difference encoding yields the following result, showing that
\ourtree{s} yield a compact parallel representation of ordered sets~\cite{BB04}.

\begin{corollary}
Given any set from $U = \{0, \ldots, m-1\}$ with $|S| = n$, the total
space of a \ourtree{} $\mathcal{\mb{\ourname}}(\alpha,
B,\mathcal{C}_{DE})$  maintaining $S$ is $O(n \log \frac{n+m}{n})$
bits for $B = \Omega(\log n)$.
\end{corollary}

\section{Algorithms}
\label{sec:algo}

\newdimen\zzsize
\zzsize=9pt
\newdimen\kwsize
\kwsize=9pt

\newcommand{\basicstyle}{\fontsize{\zzsize}{1\zzsize}\ttfamily}
\newcommand{\keywordstyle}{\fontsize{\kwsize}{1\kwsize}\ttfamily\bf}

\newdimen\zzlstwidth
\settowidth{\zzlstwidth}{{\basicstyle~}}
\newcommand{\lcm}{}

\lstset{
  xleftmargin=5.0ex,
  basewidth=\zzlstwidth,
  basicstyle=\basicstyle,
  columns=fullflexible,
  captionpos=b,
  numbers=left, numberstyle=\small, numbersep=8pt,
  language=C++,
  keywordstyle=\keywordstyle,
  keywords={nil,return,signature,sig,structure,struct,fun,fn,case,type,datatype,and,let,fn,in,end,functor,alloc,if,then,else,while,with,and,start,do,parallel},
  commentstyle=\rmfamily\slshape,
  morecomment=[l]{\%},
  lineskip={1.5pt},
  columns=fullflexible,
  keepspaces=true,
  mathescape=true,
  escapeinside={@}{@}
}

\newcommand{\ujoin}{\mbox{{join}}}
\newcommand{\uexpose}{\mbox{{expose}}}

\newcommand{\codegap}{\vspace{.05in}}

\begin{figure*}[th]
\vspace{-1em}
\begin{minipage}[t]{0.68\columnwidth}
\begin{lstlisting}
@\textbf{fold}@$(T)$ {
  flatten $T$ into array $A$
  (encoding if needed)
  return A; }@\codegap@
@\textbf{unfold}@$(A)$ {
  /* return a perfectly balanced tree
  from sorted array A */ }@\codegap@
@\textbf{expose}@$(T)$ {
  if (isflat$(T)$) {
    $T'$ = unfold($T$);
    return $(\lc(T'),k(T'), \rc(T'))$; }
  else return $(\lc(T),k(T), \rc(T))$;} @\codegap@
@\textbf{join}@($T_L,k,T_R$) {
  if (heavy($T_L, T_R$))
    return join_right$(\Tl,k,\Tr)$;
  if (heavy($T_R, T_L$))
    return join_left$(\Tl,k,\Tr)$;
  return node($T_L,k,T_R$); }@\codegap@
\end{lstlisting}
\end{minipage}
\begin{minipage}[t]{0.72\columnwidth}
\StartLineAt{19}
\begin{lstlisting}
/* join_left is symmetric */
@\textbf{join\_right}@($T_L,k,T_R$) {
  ($l,k',c$)=expose($T_L$);
  if (balance($|T_L|,|T_R|$)
    return $\node(\Tl,k,\Tr))$;
  $T'$ = join_right$(c,k,\Tr)$;
  $(l_1,k_1,r_1)$ = expose$(T')$;
  if (balance$(|l|,|T'|)$)
    return $\node$($l, k', T'$);
  if ((balanced$(|l|,|l_1|)$) and
        (balanced$(|l|+|l_1|,r_1)$))
    return rotateleft($\node$($l, k', T'$));
  else return rotateleft($\node$($l,k'$,
         rotateright$(T')$)); }@\codegap@
@\textbf{join2}@($\Tl$,$\Tr$) {
  if ($\Tl$ = nil) return $\Tr$;
  $(\Tl',m,\_)$ = split($\Tl$,last($\Tl$));
  return $\ujoin(\Tl',m,\Tr)$;  }
\end{lstlisting}
\end{minipage}
\begin{minipage}[t]{.68\columnwidth}
\StartLineAt{37}
\begin{lstlisting}
@\textbf{node}@($l,k,r$) {
  /* create node $x$ with left subtree $l$,
  root key $k$ and right subtree $r$ */
  if ($|x|>4B$) return $x$;
  if ($B\le |x| \le 2B$) return @fold@$(x)$;
  else { // $2B< |x|\le 4B$
    /* redistribute $x$'s both subtrees to
      be flat nodes with $|x|/2$ entries */
    return $x$;}} @\codegap@
@\textbf{split}@$(T,k)$ {
  if ($|T|=0$) return (nil,nil,nil);
  $(L,m,R)$ = $\uexpose(T)$;@\label{line:splitexpose}@
  if ($k == k(m)$) return ($L$,m,$R$);
  if ($k < k(m)$) {
    $(L_L,b,L_R)$ = split$(L,k)$;
    return $(L_L,b,\ujoin(L_R,m,R))$;
  } else {
    $(R_L,b,R_R)$ = split$(R,k)$;
    return $(\ujoin(L,m,R_L),b,R_R)$; } } @\codegap@
\end{lstlisting}
\end{minipage}
\caption{\small \textbf{Primitives on
    \ourtree{s}.}\label{fig:mainalgo}  All codes are functional
  (e.g. rotates copy nodes).}
\end{figure*}

\ifx\confversion\undefined
\begin{figure*}
\vspace{-1em}
\small
\begin{minipage}[t]{0.75\columnwidth}
\begin{lstlisting}
@\textbf{from\_sorted}@(A,n) {
  if ($n=0$) return nil;
  if ($n=1$) return node(nil,A[0],nil);
  $L$ = from_sorted(A,n/2) ||
    $R$ = from_sorted(A+n/2,n-n/2);
  return @\node@(L,A[n/2],R); }
@\textbf{build}@(A,n) {
  parallel_sort(A,n);
  return from_sorted(A,n); }
\end{lstlisting}
\end{minipage}
\begin{minipage}[t]{0.6\columnwidth}
\StartLineAt{10}
\begin{lstlisting}
// keep a key in $T$ only when it satisfies $f$
@\textbf{filter}@($T$,$f$) {
  if ($T$ == nil) return nil;
  ($L$,$k$,$R$) = $\uexpose(T)$;
  $\Tl$ = filter($L$,$f$) ||
      $\Tr$ = filter($R$,$f$);
  if ($f(k)$)
    return $\ujoin$($\Tl$,$k$,$\Tr$);
  else return join2($\Tl$,$\Tr$); } @\vspace{.1in}@
\end{lstlisting}
\end{minipage}
\caption{\small \textbf{Examples of parallel algorithms on
    \ourtree{s}. ``\texttt{||}'' indicates calls that are made in parallel.} \label{fig:otheralgotext}}
\end{figure*}
\else

\vspace{-1em}
\small
\begin{minipage}[t]{0.75\columnwidth}
\begin{lstlisting}
@\textbf{from\_sorted}@(A,n) {
  if ($n=0$) return nil;
  if ($n=1$) return node(nil,A[0],nil);
  $L$ = from_sorted(A,n/2) ||
    $R$ = from_sorted(A+n/2,n-n/2);
  return @\node@(L,A[n/2],R); }
@\textbf{build}@(A,n) {
  parallel_sort(A,n);
  return from_sorted(A,n); }
\end{lstlisting}
\end{minipage}
\begin{minipage}[t]{0.7\columnwidth}
\StartLineAt{10}
\begin{lstlisting}
@\textbf{union}@($T_1$,$T_2$) {
  if ($T_1$ == nil) return $T_2$;
  if ($T_2$ == nil) return $T_1$;
  ($L_2$,$k_2$,$R_2$) = $\uexpose(T_2$);@\label{line:mainunionexpose}@
  ($L_1$,$b$,$R_1$) = split($T_1$,$k_2$);@\label{line:mainunionsplit}@
  $\Tl$ = union($L_1$,$L_2$) ||
       $\Tr$ = union($R_1$,$R_2$);
  return $\ujoin$($\Tl$,$k_2$,$\Tr$); }@\label{line:mainunionjoin}
\end{lstlisting}
\end{minipage}
\begin{minipage}[t]{0.6\columnwidth}
\StartLineAt{18}
\begin{lstlisting}
// keep a key in $T$ only when it satisfies $f$
@\textbf{filter}@($T$,$f$) {
  if ($T$ == nil) return nil;
  ($L$,$k$,$R$) = $\uexpose(T)$;
  $\Tl$ = filter($L$,$f$) ||
      $\Tr$ = filter($R$,$f$);
  if ($f(k)$)
    return $\ujoin$($\Tl$,$k$,$\Tr$);
  else return join2($\Tl$,$\Tr$); } @\vspace{.1in}@
\end{lstlisting}
\end{minipage}
\caption{\small \textbf{Examples of parallel algorithms on
    \ourtree{s}. ``\texttt{||}'' indicates calls that are made in parallel.} \label{fig:otheralgotext}}
\end{figure*}

\fi

\begin{figure*}
  \vspace{-1em}
  \centering
  \includegraphics[width=2.1\columnwidth]{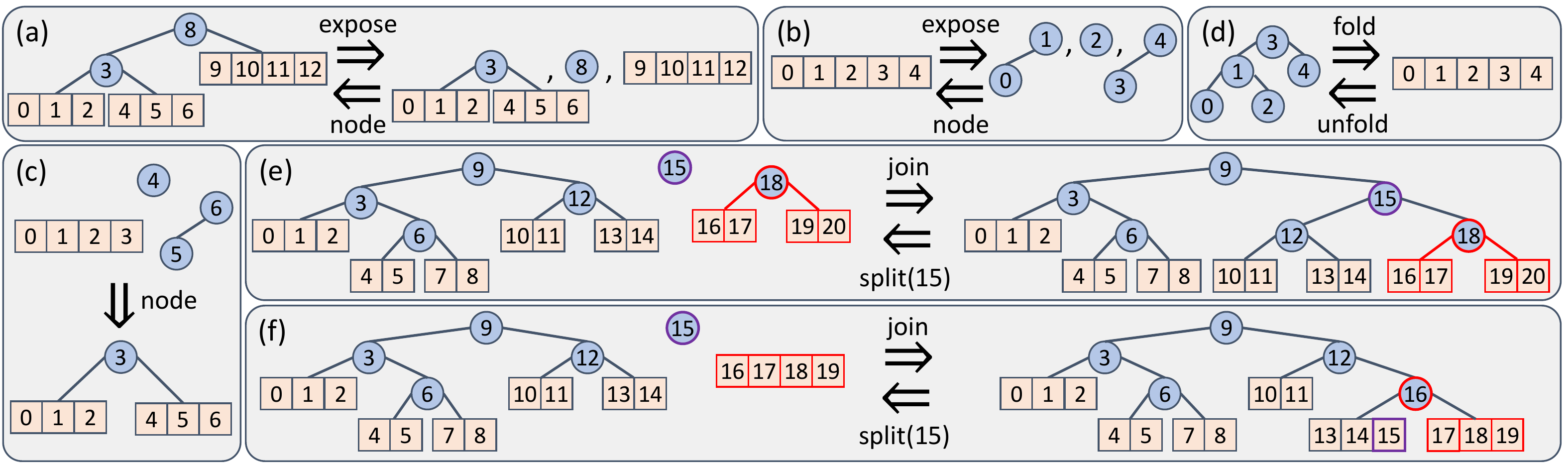}
  \caption{\small \textbf{Illustration of primitives on \ourtree{s}.} For Figures (a)--(d), $B=3$. For Figures (e)--(f), $B=2$.
   Fig.\ (a): the \expose{} function on a \regularnode{} and the \node{} function to obtain a \regularnode{} when the output tree size is larger than $4B$.
   Fig.\ (b): the \expose{} function on a \blockednode{} and the \node{} function to obtain a \blockednode{} when the output tree weight is between $B$ and $2B$.
   Fig.\ (c): the \node{} function to obtain a \blockednode{} when the output size is between $2B$ and $4B$.
   Fig.\ (d): \fold{} and \unfold{} functions.
   Fig.\ (e): \join{} function on two \regularnode{s} and its corresponding \tsplit{} function.
   Fig.\ (f): \join{} function on a \regularnode{} and a \blockednode{} and its corresponding \tsplit{} function.
   }\label{fig:primitive}
\end{figure*}

\hide{
\begin{figure*}
  \centering
  \includegraphics[width=1.8\columnwidth]{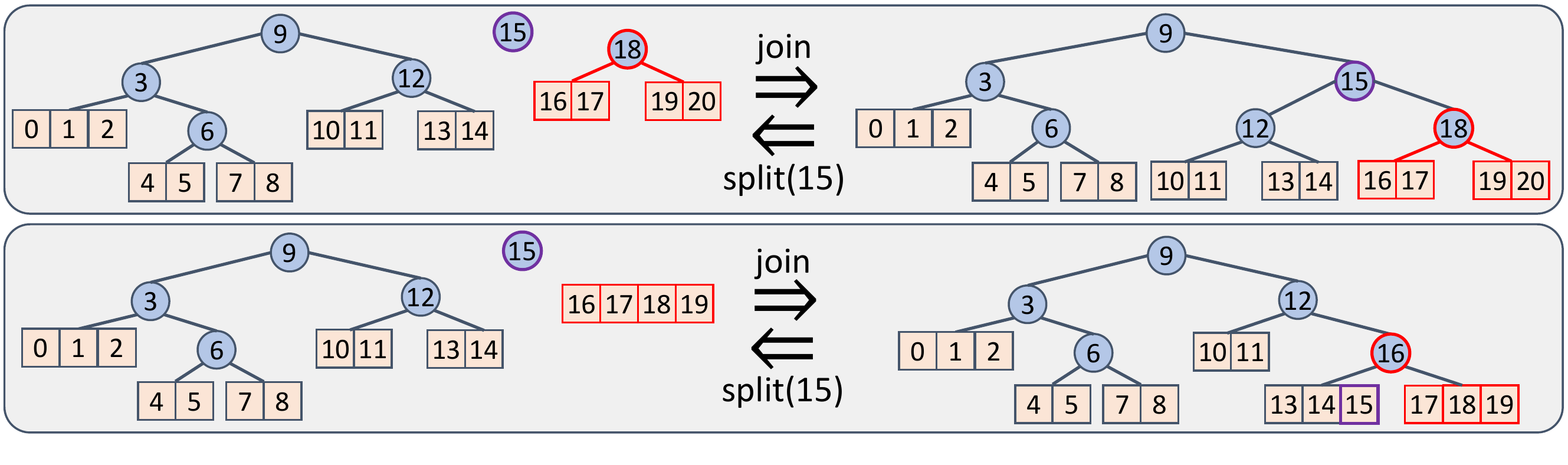}
  \caption{-}\label{fig:joinsplit}
\end{figure*}
}

We now describe \join-based algorithms on \ourtree{s}.
To enable a general ordered map interface, we implement \ourtree{s} based
on the \pam{} interface. PAM supports dozens of operations on sequences, sets,
maps, and augmented maps, and it would require significant work to
re-implement them all.
Instead, we carefully redesigned \join{} and \expose{} such that
\emph{all the other algorithms can remain the same as in \pam{}}.
In particular, none of the other algorithms have to deal with the
\block{ed} leaves or compression, which greatly simplifies the
algorithm design and correctness arguments.
We found that the overhead of this approach is not large, but for many frequently-used
operations, we design special base cases for dealing with compressed
nodes. These base cases can improve the performance by up to 6x (see
\cref{sec:impl}).
Some of the theoretical results also require special base cases (see
\cref{sec:setproof}).

At a high-level, when exposing a \blockednode{}, the node is
automatically expanded (using \unfold), and similarly when \join{}
obtains a \composite{} tree of size $B$ to $2B$, it is flattened
(\fold{}).
An illustration of \unfold{} and \fold{} is shown in
\cref{fig:primitive}.
We start with the \join{} and \expose{} algorithms. We then present
the \union{} algorithm as an example to illustrate \join-based
algorithms, and give the code for other functions in
\cref{fig:otheralgotext} 
\ifx\confversion\undefined
and \cref{fig:setcode}.
\else
and the Appendix.
\fi
We focus on \union{} as it is the core sub-routine used in
applications such as inserting or deleting batches of vertices and
edges in graphs, combining inner trees when constructing range trees,
and updating sets of documents in an inverted index, among others.

\myparagraph{Expose.} This function returns the left
subtree, root data and the right subtree of a node $T$. For a \regularnode{},
this function just reads the child pointers and the root. For a \blockednode{}, this function first \unfold{}s the
tree into a perfectly balanced tree and then reads the corresponding
data.

\myparagraph{Join.} Recall that the \join{} function takes two trees
$T_L$ and $T_R$, and a key $k$ (or a key-value) as input, and returns
a balanced tree concatenating entries in $T_L$, $k$ and $T_R$ in order
(see \cref{fig:primitive}).
In other words, when trees are used for ordered sets or maps, $k$
should be larger than all keys in $T_L$ and smaller than all keys
in $T_R$.
Pseudocode for \join{} is shown in \cref{fig:mainalgo}.

The algorithm first compares the weights of $T_L$ and $T_R$.
When balanced, they are directly connected
by $k$. The other two cases are symmetric so WLOG we assume
$|T_L|>|T_R|$. In this case, the algorithm must attach $T_R$
in the right spine of $T_L$, which will be handled by
\myfunc{join\_right}$(T_L,k,T_R)$. This algorithm first checks
if $T_L$ and $T_R$ are balanced and connects them if so.
Otherwise, it recursively calls \myfunc{join\_right} on $\rc(T_L)$ and $T_R$,
getting $T'$. If we re-attach $T'$ as $T_L$'s right
child, we will get a ``correct'' output tree (modulo balance).  We then use a single or
double rotation to rebalance if necessary. It is known that either a
single or double rotation can rebalance a weight-balanced tree in this
situation \cite{blelloch2016just}. This guarantees the \emph{weight balance} invariant of \ourtree{s}.

To also guarantee the \emph{blocked leaves} invariant, we add two conditions when calling \node{} to create a new node with its left and right subtrees.
Whenever a node with size $B$ to $2B$ is created, we \fold{} the tree into a \blockednode{}.
Whenever a node with size $2B$ to $4B$ is created, we extract the median of the tree as the root to re-distribute its two subtrees, such that both subtrees
are \blockednode{s} with (almost) the same size.

\begin{lemma}
\label{lem:joincorrect}
The \join{} function maintains the invariants of \ourtree{s}.
\end{lemma}

\myparagraph{Split.}
For a \ourtree{} $T$ and key $k$, \tsplit{}$(T,k)$ returns a triple
$(\Tl,b,\Tr)$, where $\Tl$ ($\Tr$) is a tree containing all keys in
$T$ that are less (greater) than $k$, and $b$ the entry of key $k$ if $k \in T$ (see \cref{fig:primitive}).  We first use \expose$(T)$ to get its left (right) subtrees $\lc(T)$ ($\rc(T)$) and root key $k(T)$, and compare $k$ with $k(T)$. If $k=k(T)$, we simply return $(\lc(T),k, \rc(T))$. Otherwise WLOG we assume $k$ is smaller. In that case, the entire right subtree $\rc(T)$ and the root $k(T)$ belong to $\Tr$. We then split $\lc(T)$ by $k$, getting $(L_L,b,L_R)$. By definition, all keys smaller than $k$ should be in $L_L$, and all keys larger than $k$ can be obtained by $\join(L_R,k(T),\rc(T))$.
\hide{
The algorithm first searches for $k$ in $T$, splitting the tree along the path into three parts: keys to the left of the path, $k$ itself (if it exists), and keys to the right.  Then by applying \join{}, the algorithm merges all the subtrees on the left side (using keys on the path as intermediate nodes) from bottom to top to form $\Tl$, and merges the right parts to form $\Tr$.
}

\hide{
\myparagraph{Join2.}
\joinTwo{}$(\Tl,\Tr)$ returns a binary tree for which the in-order
values is the concatenation of the in-order values of the binary trees
$\Tl$ and $\Tr$ (the same as \join{} but without the middle key).
For BSTs, all keys in $\Tl$ have to be less than keys in $\Tr$.
\joinTwo{} first finds the last element $k$ (by following the right spine) in $\Tl$ and
on the way back to root, joins the subtrees along the path, which is
similar to \tsplit{}
$\Tl$ by $k$. We denote the result of dropping $k$ in $T_L$ as $\Tl'$. Then \join($\Tl',k,\Tr$)
is the result of \joinTwo.
}

\myparagraph{Union.} Using \texttt{join} and \texttt{split}, we can implement set algorithms on two \ourtree{s}, such as \union{}, \intersection{} and \difference{}.
We describe \union{} as an example (the other two are similar).
This algorithm uses divide-and-conquer.
At each level of recursion,
$T_1$ is split by the root of $T_2$, breaking $T_1$ into two subsets with all keys smaller (larger) than $k(T_2)$, denoted as $L_1$ ($R_1)$.
Then two recursive calls to \union{} are made in parallel.  One unions
$L(T_2)$ with $L_1$ (all keys smaller than $k(T_2)$), returning $\Tl$, and the other one unions
$R(T_2)$ with $R_1$ (all keys larger than $k(T_2)$), returning $\Tr$.
Finally the algorithm combines the results with \join$(\Tl,k(T_2),\Tr)$.

\hide{
\myparagraph{\union{}.} We can also implement set algorithms on two
\ourtree{s}, such as \union{}, \intersection{} and \difference{}.
We will describe \union{} as an example, and the other two are
similar.
Pseudocode for \union{} is shown in \cref{fig:joinandunion}.
This function needs a helper function \tsplit{}$(T,k)$, which
returns a triple $(\Tl,b,\Tr)$, where $\Tl$ ($\Tr$) is a tree
containing all keys in $T$ that are less (larger) than $k$, and $b$ is
a flag indicating whether $k \in T$ (or its value if exists).
The implementation of \tsplit{} is also based on \join{} and is the
same as that in \pam{}. We present its pseudocode in the appendix.
The \union{} algorithm uses divide-and-conquer.
At each level of recursion, $T_1$ is split by the root of $T_2$,
breaking $T_1$ into three parts: one with all keys smaller than
$k(T_2)$ (denoted as $L_1$), one in the middle either of only one key
equal to $k(T_2)$ (when $k(T_2)\in T_1$) or empty (when $k(T_2) \notin
T_1$), and the third one with all keys larger than $k(T_2)$ (denoted
as $R_1$).
Then two recursive calls to \union{} are made in parallel.  One unions
$L(T_2)$ with $L_1$, returning $\Tl$, and the other one unions
$R(T_2)$ with $R_1$, returning $\Tr$.  Finally the algorithm combines the three parts by
\join$(\Tl,k(T_2),\Tr)$, which is valid since $k(T_2)$ is greater than
all keys in $\Tl$ are less than all keys in $\Tr$.
}

\ifx\confversion\undefined

\ifx\confversion\undefined

\begin{figure*}
\small
\begin{minipage}[t]{0.78\columnwidth}
\begin{lstlisting}
@\textbf{m\_ins\_helper}@($T$,A,m) {
  if ($T$ == nil) return from_sorted(A,m);
  if ($n$ == 0) return $T$;
  ($L$,$k$,$R$) = @\uexpose@($T$);
  $s$ = binary_search(A,m,k);
  if ($A[s]$ == k) $b$ = 1;
  $\Tl$ = m_ins_helper($L$,A,s); ||
    $\Tr$ = m_ins_helper($R$,A+s+b,m-s-b);
  return @\ujoin@($\Tl$,$k$,$\Tr$); }@\codegap@
@\textbf{multi\_insert}@($T$,A) {
  $A'$ = parallel_sort(A);
  return m_ins_helper($T$,$A'$,$|A|$); }
\end{lstlisting}
\end{minipage}
\begin{minipage}[t]{0.65\columnwidth}
\StartLineAt{13}
\begin{lstlisting}
// $f$ is an associative binary operation
// $I$ is the identity
@\textbf{reduce}@($T$,$f$,$I$) {
  if ($T$ is nil) return $I$;
  ($L$,$k$,$R$) = $\uexpose(T)$;
  $x$ = reduce($L$,$f$,$I$)
      || $y$ = reduce($R$,$f$,$I$);
  return $f(f(x,k),y)$;
\end{lstlisting}
\end{minipage}
\begin{minipage}[t]{0.65\columnwidth}
\StartLineAt{21}
\begin{lstlisting}
// map each entry in $T$ using function $f$
@\textbf{map}@($T$,$f$) {
  if ($T$ == nil) return nil;
  ($L$,$k$,$R$) = $\uexpose(T)$;
  $\Tl$ = map($L$,$f$)
      || $\Tr$ = map($R$,$f$);
  return $\ujoin$($\Tl$,$k$,$\Tr$);
\end{lstlisting}
\end{minipage}
\caption{Some other algorithms on \ourtree{s}. \label{fig:otheralgo}}
\end{figure*}

\fi

\begin{figure*}[th]
\begin{minipage}[t]{1.1\columnwidth}
\begin{lstlisting}
@\textbf{join}@($T_L,k,T_R$,expand=false) {
  if (heavy($T_L, T_R$)) return join_right($\Tl,k,\Tr,$expand);
  if (heavy($T_R, T_L$)) return join_left($\Tl,k,\Tr,$expand);
  return node($T_L,k,T_R$); } @\vspace{.1in}@
@\textbf{join\_right}@($T_L,k,T_R$,expand=false) {
  ($l,k',c$)=expose($T_L$);
  if (balance($|T_L|,|T_R|$)) return $\node$($\Tl,k,\Tr,$expand);
  $T'$ = join_right$(c,k,\Tr)$;
  $(l_1,k_1,r_1)$ = expose$(T')$;
  if (balanced$(|l|,|T'|)$) return $\node$($l, k', T',$expand);
  if ((balanced$(|l|,|l_1|)$) and (balanced$(|l|+|l_1|,r_1)$))
    rotate_left($\node$($l, k', T'$));
  else rotate_left($\node$($l,k'$,rotate_right$(T')$)); } @\vspace{.1in}@
@\textbf{join\_left}@($T_L,k,T_R$) { /* symmetric to join_right */ }@\vspace{.1in}@
@\textbf{node}@($l,k,r,$expand=false) {
  let the left child of $k$ be $l$;
  let the right child of $k$ be $r$;
  if (expand) return $k$;
  if ($|k|>4B$) return $k$;
  if ($B\le |k| \le 2B$) return @fold@$(k)$;
  else { // $2B< |k|\le 4B$
    /* redistribute $k$'s both subtrees to be flat nodes with $|k|/2$ entries */
    return $k$;}  }
\end{lstlisting}
\end{minipage}
\hfill
\begin{minipage}[t]{.8\columnwidth}
\StartLineAt{24}
\begin{lstlisting}
@\textbf{expose}@$(T)$ {
  if (isregular$(T)$) return $(\lc(T),k(T), \rc(T))$;
  else {
    $T'$ = unfold($T$);
    return $(\lc(T'),k(T'), \rc(T'))$; } } @\vspace{.1in}@
@\textbf{fold}@$(T)$ {
  if ($B\le w(T)\le 2B$) {
    flatten $T$ into array $A$
    (encoding if needed)
    return A;
  } else return $T$; }@\vspace{.1in}@
@\textbf{unfold}@$(A)$ {
  /* build a perfectly balanced tree $T$
  from entries in (sorted) array A */ }@\vspace{.1in}@
@\textbf{refold}@$(T)$ {
  if ($T$ is not marked) return $T$;
  if ($B\le|T|\le 2B$) return @\fold@$(T)$;
  else {
    $(L,m,R)=\expose(T)$;
    $\Tl=\refold(L)$ || $\Tr=\refold(R)$;
    return $\ujoin(\Tl,m,\Tr)$; }}
\end{lstlisting}
\end{minipage}
\caption{Some useful primitives for \join{}-based algorithms on \ourtree{s} for more efficient set algorithms. \label{fig:codeappjoin}}
\end{figure*}

\ifx\confversion\undefined
\begin{figure*}
\small
\hfill
\begin{minipage}[t]{\columnwidth}
\begin{lstlisting}
@\textbf{splitLast}@($T$) {
  $(L,k,R)$ = $\uexpose(T)$;
  if ($R$ == nil) return ($L,k$);
  else {
    $(T',k')$ = splitLast($R$);
    return $(\ujoin(L,k,T'),k')$; } }
\end{lstlisting}
\end{minipage}
\hfill
\begin{minipage}[t]{\columnwidth}
\StartLineAt{7}
\begin{lstlisting}
@\textbf{join2}@($\Tl$,$\Tr$) {
  if ($\Tl$ = nil) return $\Tr$;
  else {
    $(\Tl',k)$ = splitLast($\Tl$);
    return $\ujoin(\Tl',k,\Tr)$; } }@\vspace{.1in}@
\end{lstlisting}
\end{minipage}
\hfill
\textbf{Simple version used in implementation:}\\
\begin{minipage}[t]{0.65\columnwidth}
\StartLineAt{12}
\begin{lstlisting}
@\textbf{union}@($T_1$,$T_2$) {
  if ($T_1$ == nil) return $T_2$;
  if ($T_2$ == nil) return $T_1$;
  ($L_2$,$k_2$,$R_2$) = $\uexpose(T_2$);@\label{line:mainunionexpose}@
  ($L_1$,$b$,$R_1$) = split($T_1$,$k_2$);@\label{line:mainunionsplit}@
  $\Tl$ = union($L_1$,$L_2$)
      || $\Tr$ = union($R_1$,$R_2$);
  return $\ujoin$($\Tl$,$k_2$,$\Tr$); }@\label{line:mainunionjoin}
%  $\ujoin$($\Tl$,$k_2$,$\Tr$) }@\vspace{.1in}@
\end{lstlisting}
\end{minipage}
\begin{minipage}[t]{0.65\columnwidth}
\StartLineAt{20}
\begin{lstlisting}
@\textbf{intersect}@($T_1$,$T_2$) {
  if ($T_1$ == nil) return nil;
  if ($T_2$ == nil) return nil;
  ($L_2$,$k_2$,$R_2$) = $\uexpose(T_2)$;
  ($L_1$,$b$,$R_1$) = split($T_1$,$k_2$);
  $\Tl$ = intersect($L_1$,$L_2$)
      || $\Tr$ = intersect($R_1$,$R_2$);
  if ($b$) return $\ujoin$($\Tl$,$k_2$,$\Tr$);
  else return join2($\Tl$,$\Tr$); } @\vspace{.1in}@
\end{lstlisting}
\end{minipage}
\begin{minipage}[t]{0.65\columnwidth}
\StartLineAt{29}
\begin{lstlisting}
@\textbf{difference}@($T_1$,$T_2$) {
  if ($T_1$ == nil) return nil;
  if ($T_2$ == nil) return $T_1$;
  ($L_2$,$k_2$,$R_2$) = expose($T_2$);
  ($L_1$,$b$,$R_1$) = split($T_1$,$k_2$);
  $\Tl$ = difference($L_1$,$L_2$)
      || $\Tr$ = difference($R_1$,$R_2$);
  return join2($\Tl$,$\Tr$); }
\end{lstlisting}
\end{minipage}\\
\textbf{Special base cases for tighter bound:}\\
\begin{minipage}[t]{\columnwidth}
\StartLineAt{37}
\begin{lstlisting}
@\textbf{union\_}@($T_1$,$T_2$) {
  if ($T_1$ == nil) return $T_2$;
  if ($T_2$ == nil) return $T_1$;
  if (isflat($T_1$) or isflat($T_2$)) {
    return refold(union_base($T_1,T_2$));}@\label{union:refold}@
  ($L_2$,$k_2$,$R_2$) = $\uexpose(T_2$);
  ($L_1$,$b$,$R_1$) = split($T_1$,$k_2$);
  $\Tl$ = union_($L_1$,$L_2$) || $\Tr$ = union_($R_1$,$R_2$);
  return $\ujoin$($\Tl$,$k_2$,$\Tr$) }@\vspace{.1in}@
\end{lstlisting}
\end{minipage}
\StartLineAt{46}
\begin{minipage}[t]{\columnwidth}
\begin{lstlisting}
@\textbf{union\_base}@($T_1$,$T_2$) {
  if ($T_1$ == nil) return $T_2$;
  if ($T_2$ == nil) return $T_1$;
  if (isflat($T_1$)) unfold($T_1$);
  if (isflat($T_2$)) unfold($T_2$);
  ($L_2$,$k_2$,$R_2$) = $\uexpose(T_2$);
  ($L_1$,$b$,$R_1$) = split($T_1$,$k_2$,true);
  $\Tl$ = union_base($L_1$,$L_2$) || $\Tr$ = union_base($R_1$,$R_2$);
  return join($\Tl$,$k_2$,$\Tr$,true); }
\end{lstlisting}
\end{minipage}
\caption{Set algorithms on \ourtree{s}.
\label{fig:setcode}}
\end{figure*} 
\fi

\myparagraph{Other algorithms.}
We present the pseudocode for the other two set algorithms
(\intersection{} and \difference{}) in \cref{fig:setcode}. We also
show the code for three other useful functions, \multiinsert{}, \texttt{map}
and \texttt{reduce} in \cref{fig:otheralgo}. We note that these algorithms are
exactly the same as in \pam{}, by extracting out the semantics of
\join{} and \expose{}.

\else
\myparagraph{Other algorithms.} We show the pseudocode of other
parallel algorithms in \cref{fig:otheralgotext} and more in the
full version of the paper. We omit the details as they are
self-explanatory and all of them are exactly the same as in \pam{},
just by plugging in the new version of \join{} and \expose{} functions
for \ourtree{}. Almost all of them use divide-and-conquer to enable
parallelism.  We refer the reader to \cite{pam} for more details.
\fi

\medskip

Importantly, all of our \ourtree{} algorithms are theoretically
efficient. We present the work-span bound in
\cref{table:algorithmcosts} and give a proof for \union{} as an
example in \cref{sec:theory}. Note that \cref{lem:joincorrect}
ensures the correctness of the other algorithms, as their return
values are always obtained by a \join{}.

\begin{theorem}
  All \join-based algorithms on \ourtree{} maintains the invariants of \ourtree{s}.
\end{theorem}

\section{Theoretical Guarantees}
\label{sec:theory}

In the following section we show work and span bounds for operations
on \ourtree{s}. We assume the encoding scheme is empty, which means
that to flatten or \explode{} a block of size $n$ costs $O(n)$ work
and $O(\log n)$ span.
If the encoding scheme is not parallelizable (e.g., for difference
encoding), the span bound of the algorithms will be affected. We
present more details in 
\ifx\confversion\undefined
Section~\ref{app:setspan}.
\else
the full version of the paper.
\fi

We start with the cost of the \join{} and \tsplit{} algorithms.

\begin{restatable}{theorem}{joincost}
\label{lem:joincost}
Consider a \join{} algorithm on two \ourtree{s} $\Tl$, $\Tr$ and an key $k$. Let $n=\max(|\Tl|,|\Tr|)$ and $m=\min(|\Tl|,|\Tr|)$.
If both $\Tl$ and $\Tr$ are \composite{} trees, the algorithm takes $O\left(\log\frac{n}{m}\right)$ work and span.
If both $\Tl$ and $\Tr$ are \simple{} trees, the algorithm takes $O(\B)$ work and $O(\log B)$ span.
Otherwise, the algorithm takes $O\left(\B + n/B\right)$ work and $O(\log n)$ span.
\end{restatable}

\ifx\confversion\undefined
\begin{proof}
WLOG, let's assume $n=|\Tl|\ge|\Tr|=m$. 

For two \composite{} trees, we first prove that \join{} never decompresses a leaf.
Note that the algorithm will follow the right spine of the tree until finding a subtree $t$ in $\Tl$ that is balanced with $\Tr$, we will prove that there exist a \internalnode{} $t$ that is balanced with $\Tr$.
This is because as a \composite{} tree, $\Tr$ has size at least $2\B$. Along the right spine of $\Tl$, the smallest \composite{} subtree has size at most $4\B$. For any $\alpha\le 1/3$, we must be able to find a \composite{} subtree in $\Tl$ that is balanced with $\Tr$. This proves that the total number of tree nodes we need to visit on the right spine is $O\left(\log\frac{n}{m}\right)$.

For two \simple{} trees, the work is no more than copying both $\Tl$ and $\Tr$ and concatenating them, which is $O(\B)$ work and $O(\log B)$ span.

If $\Tl$ is a \composite{} tree and $\Tr$ is a \simple{} tree, we need to first follow the right spine to find a \leafnode{} $l$ in $\Tl$, which takes $O\left(\log\frac{n}{B}\right)$ time. Then it combines the \leafnode{} with $\Tr$, which flattens both $l$ ant $\Tr$, concatenates them, and rebalance the result. This process takes no more than $O(B)$ work and $O(\log B)$ span.
\end{proof}
\fi

\begin{restatable}{theorem}{splitcost}
\label{thm:splitcost}
  Consider a \tsplit{} algorithm on a \ourtree{} $T$. If $T$ is a \composite{} tree, the work and span of \tsplit{} are $O(\log \frac{|T|}{\B}+\B)$ and $O(\log |T|)$, respectively.
  If $T$ is an \simple{} tree, the work and span of \tsplit{} is $O(\log |T|)$.
\end{restatable}

\ifx\confversion\undefined
\begin{proof}
  For a \simple{} tree the cost directly follows the result on \ptree{s}~\cite{blelloch2016just}. For a \composite{} tree, the only difference of \tsplit{} on \ourtree{s} from \tsplit{} on PAM trees is the \unfold{} performed in \expose{} and  \fold{} in \join{}. It takes $O(\log |T|/B)$ steps to reach a \leafnode{}. The \unfold{} function in \expose{} (Line \ref{line:splitexpose}) is performed at most once. For all \join{} calls, at most one on each side can involve a \simple{} tree. So the total overhead is at most $O(\B)$ in work and $O(\log B)$ in span.
\end{proof}
\else
Due to page limit, we provide the proofs of \cref{lem:joincost,thm:splitcost} in the full version of this paper.
\fi
Based on these results, we now analyze the cost of the set operations.


\begin{theorem}
\label{thm:unionwork}
Consider the \union{} algorithm (and the closely related
  \intersection{} and \difference{} algorithms
\ifx\confversion\undefined
) in \cref{fig:setcode}
\else
  in the Appendix) in \cref{fig:mainalgo}
\fi
 on two \ourtree{s} of sizes $m$ and $n\ge m$.
  The work and span for these algorithms are $O\left( m\log
  \frac{n}{m} + m\B\right)$ and $O(\log n \log m)$ respectively.
\end{theorem}

To prove the theorem, we first present some definitions and lemmas. First, note that all the work can be asymptotically bounded by the three categories below:

\begin{enumerate}[label=(\arabic*).]
  \item \emp{split work}: all work done by \tsplit{} (Line \ref{line:mainunionsplit}),
  \item \emp{join work}: all work done by \join{} (Line \ref{line:mainunionjoin}) or \joinTwo{} in \intersection{} and \difference,
  \item \emp{expose work}: all work done by \expose{} (Line \ref{line:mainunionexpose}).
\end{enumerate}

One observation is that the split work is identical among the three
set algorithms. This is because the three algorithms behave the same
on the way down the recursion when doing \tsplit{}s, and only differ
in what they do at the base case and on the way up the recursion when
building the output tree (see the other two set algorithms in \cref{fig:setcode}). 

We use \op{} to denote the set operation (one of \union{},
\intersection{} or \difference{}).
In these algorithms, the tree $T_1$ is split by the keys in $T_2$.  We
call $T_1$ the \emp{decomposed tree} and $T_2$ the
\emp{pivot tree}, denoted as $T_d$ and $T_p$ respectively.
Let $m=\min(|T_p|, |T_d|)$ and $n=\max(|T_p|,|T_d|)$.
\hide{
We denote the subtree rooted at $v \in T_p$ as $T_p(v)$, and the tree
of keys from $T_d$ that $v$ is operated with as $T_d(v)$ (i.e.,
\op$(v,T_d(v))$ is called at some point in the algorithm.
This essentially means that $v$'s subtree in $T_p$ is processed with the tree $T_d(v)$ in a recursive call.
Note that $T_d(v)$ may not be a subtree in $T_d$, but is a tree of a subset of keys in $T_d$. We call such $T_d(v)$ a \emp{\subsettree{}} of the \decomposedtree{} $T_d$. We say $v\in T_p$ \emp{\dealswith} $T_d(v)$.
For $v \in T_p$, we refer to $|T_d(v)|$ as its \emp{splitting size}.
}

\begin{restatable}{lemma}{joinlarger}\label{lem:joinlarger}
For each function call to \op{} on
trees $P\subseteq T_p$ and $D\subseteq T_d$, the work done by \join{} (or \joinTwo{})
is asymptotically bounded by the work done by \tsplit{}.
\end{restatable}

\ifx\confversion\undefined
\begin{proof}

Assume the return value is $R$.

First of all, the work of \tsplit{} is $\Theta(\log |D|+\B)$. Note that the work of \join{} (or \joinTwo{}) can be bounded by $O(\log|R|+\B)$.

Notice that \difference{} returns the keys in $D\backslash P$. 
Thus for both \intersection{} and \difference{} we have $R\subseteq D$. Therefore $|R|\le |D|$, which means the work done by \join{} or \joinTwo{} is no more than the work done by \tsplit{}.

For \union{}, first of all, we always call \join{} instead of \joinTwo{}. If $|P|\le |D|$, then $|R|\le 2|D|$.
\join{} costs work $O(\log |R|+B)=O(\log |P|+B)$, which is no more than $\Theta(\log |D|+\B)$.

Consider $|P|>|D|$. The subtrees $\lc(P)$ and $\rc(P)$,
which are used in the recursive calls, have size at least $\alpha|P|$ and at most $(1-\alpha)|P|$.
After combining with a subset of elements in $D$ (which has size smaller than $|P|$), the return value of each recursive call should have size at least $\alpha|P|$ and $(2-\alpha)|P|$. Denote these two trees from recursive calls as $t_l$ and $t_r$, respectively.
Note that $\alpha$ is a constant, so the difference of size between $t_l$ and $t_r$ is also no more than a constant.
WLOG assume $|t_l|\ge |t_r|$. In the following, we discuss different cases of whether $t_l$ and $t_r$ are \composite{} or \simple{} trees.
We will show that, in all cases, joining $|t_l|$ and $|t_r|$ has work $O(\log |D|+\B)$.

\begin{enumerate}
  \item When both $t_l$ and $t_r$ are \simple{} trees. From \cref{lem:joincost}, \join{} costs $O(\B)$ work.
  \item When both $t_l$ and $t_r$ are \composite{} trees. From \cref{lem:joincost}, \join{} costs $O(\log \frac{|r_l|}{|t_r|})=O(1)$ work.
  \item When $t_l$ is a \composite{} tree, but $t_r$ is a \simple{} tree. This means that $|t_l|>\B$ and $|t_r|\le \B$.
  From \cref{lem:joincost}, \join{} costs $O\left(\B + \log\frac{|t_l|}{\B}\right)$ work. Note that, as stated above, $|t_r|\ge \alpha|P|$. Considering $|t_r|\le B$, we know that $|P|=O(B)$, which also indicates $|t_l|=O(\B)$. Plug this into the work of \join{} $O\left(\B + \log\frac{|t_l|}{\B}\right)$, we can get the bound $O(\B)$.
\end{enumerate}

In summary, in all cases the work of \join{} or \joinTwo{} is asymptotically bounded by the corresponding \tsplit{} function.
\end{proof}
\else
We present the proof in the full version of this paper.
\fi
Next, we prove the bounds for split work and expose work, respectively.

\begin{lemma}
\label{lem:exposework}
The expose work is $O(\min(m\B,n))$.
\end{lemma}
\begin{proof}
  \expose{} costs $\Theta(\B)$ when the subtree is a \blockednode{}, and $O(1)$ otherwise. 
  At most $O(m)$ nodes in $T_p$ will split $T_d$, so the total cost is $O(mB)$. The cost is also no more than $O(n)$ since each node is involved in at most one \expose{}, after which the \blockednode{} will be fulled \exploded{}. In summary the cost is $O(\min(m\B,n))$.
\end{proof}


\hide{
\begin{lemma}
\label{lem:unfold}
The work of all \unfold{} functions in \tsplit{} is $O(m\B)$.
\end{lemma}
\begin{proof}
  Every node in $T_p$ will be used to \tsplit{} at most one \subsettree{} of $T_d$, which involves at most one \unfold{} function all with cost $O(\B)$.
  There can be at most $O(m)$ nodes in $T_p$ used to split $T_d$. Thus the total work is $O(mB)$.
\end{proof}}



\hide{
\begin{lemma}
\label{lem:PAMsplit}
The total split work done on two \exploded{} weight-balanced trees of sizes $n$ and $m\le n$ is $O\left(m\log \frac{n}{m}\right)$.
\end{lemma}
This directly follows \cite{blelloch2016just}.
}

\begin{lemma}
\label{lem:splitwork}
The total \tsplit{} work is $O\left(m\log \frac{n}{m}+m\B\right)$.
\end{lemma}
\begin{proof}
  The total split work can be viewed as two parts:
  the total work to done by \tsplit{} functions to traverse and split
  non-\compressednode{s}, and the work to expose and split the
  \compressednode{s}. Note that here ``non-\compressednode{s}''
  include both \internalnode{s} in \composite{} trees, and all the
  nodes in \exploded{} trees.

\hide{
  We first note that in the base cases, the
  \tsplit{}$(\cdot,\cdot,\true)$ function must be working on a
  \exploded{} tree. As a result, the total work to traverse and split all non-\compressednode{s} can be asymptotically bounded by the split work considering if both $T_p$ and $T_d$ are fully expanded. This cost can be computed by \cref{lem:PAMsplit}, which is $O\left(m\log \frac{n}{m}\right)$.
}

  First of all, the total work to traverse and split all
  non-\compressednode{s} can be asymptotically bounded by the split
  work when both $T_p$ and $T_d$ are considered to be fully \exploded{}. This
  cost is $O\left(m\log \frac{n}{m}\right)$ from the result for
  \ptree{s}~\cite{blelloch2016just}.


  We then consider all work done by \tsplit{} functions on
  \compressednode{s}. The only extra cost is the cost of \unfold{}.  
  Every node in $T_p$ will be used at most once to \tsplit{} $T_d$, which involves at most one \unfold{} function with cost $O(\B)$.
  There can be at most $O(m)$ nodes in $T_p$ used to split $T_d$. Thus the total unfold work in \tsplit{} is $O(mB)$.

  Therefore in total the split work is $O\left(m\log \frac{n}{m}+m\B\right)$.
%
\end{proof}

We can now prove \cref{thm:unionwork}.

\begin{proof} \textit{(\cref{thm:unionwork})}
Combining \cref{lem:splitwork,lem:exposework,lem:joinlarger} proves the work bound in \cref{thm:unionwork}.
For the span, note that the algorithms need $O(\log |T_p|/B)$ rounds to reach a \blockednode{},
where the \blockednode{} will be \exploded{}, taking $O(\log B)$ span. Then the algorithm keeps
recursing until a \nil{} node is reached, which takes $O(\log B)$ rounds. In each of the recursive calls,
we need $O(\log |T_d|)$ span to deal with \tsplit{} and \join{}. In total the span is $O(\log m\log n)$.
\end{proof}

\ifx\confversion\undefined
\subsection{Set Algorithms with Better Work Bound}\label{sec:setproof}

Note that the $O(m\B)$ term can be expensive when $m$ is large.
In fact, we can show a tighter bound using a more efficient (but more
complicated) base case, which we present next. We note that in our
implementation, we use the version in \cref{fig:mainalgo}, which has
good performance in practice. The main result in this section is the
theorem below, based on the algorithm shown in \cref{fig:setcode} as
\texttt{union\_}.

\begin{restatable}{theorem}{setefficient}
\label{thm:unionworkbetter}
There exist algorithms for \union{}, \intersection{} and \difference{}
  on two \ourtree{s} of sizes $m$ and $n$ ($n\ge m$) with work $\displaystyle O\left( m\log \frac{n}{m} + \min(n,m\B)\right)$ and span $O(\log n \log m)$.
\end{restatable}

\myparagraph{Algorithm.}
The general idea is to avoid folding and unfolding \simple{} trees
during the \union{} algorithm. In particular, we hope each
\blockednode{} is folded and unfolded $O(1)$ times during the entire
\union{} algorithm. To ensure this, we implemented a special base case for
the set algorithms. We show the code for \union{} in
\cref{fig:setcode}, the other two are similar. The base case will
explicitly determine if the current input is a \blockednode{}. If any
of them is, it will be \exploded{} directly. In the subsequent
\join{}, we will pass an extra parameter to indicate that the tree is
already \exploded{}, and thus there is no need to fold or unfold them
again in this \join{} algorithm. This parameter is the last parameter
of the \join{} algorithm in \cref{fig:codeappjoin}. It is set to
\false{} by default, which makes it exactly the same as the version in
\cref{fig:mainalgo}. When it is set to \true{}, the \join{} algorithm
will never \fold{} any node. Instead, at the end of the base case, the
entire result tree will be fixed using \refold{} (Line
\ref{union:refold}), which traverses the tree and folds any subtree of
size $B$ to $2B$ back to blocks.

\myparagraph{Theoretical Cost.}
As the span bound is not affected, we will only show the new proof for
the work here. First, note that total work for each of these
algorithms can be considered as several parts:

\begin{enumerate}[label=(\arabic*).]
  \item \label{item:unfold} all work done by \unfold{} operations in the base cases, including possibly those in \tsplit{} and \join{} (or \joinTwo{}) function calls, denoted as \emp{unfold work},
  \item all work done by \refold{} operations, denoted as \emp{refold work},
  \item all work done by \tsplit{} operations except for the \unfold{} work in base cases (already charged in \ref{item:unfold}, denoted as \emp{split work},
  \item all work done by \join{} or \joinTwo{} operations, denoted as \emp{join work}.
\end{enumerate}

We note that all the rest of the cost can be asymptotically bounded by the above four categories of work.

One observation is that the split work and unfold work are identical
among the three set algorithms. This is because the three algorithms
behave the same on the way down the recursion when doing \tsplit{}s,
and algorithms only differ in what they do at the base case and on the
way up the recursion when they join back.

We start with some notation.  We follow some the notation used in \cite{blelloch2016just}.
Throughout the section, we use \op{} to denote the algorithm or function call on \union{}, \intersection{} or \difference{}, and use \opbase{} to denote the corresponding base case algorithm or function call (using the unfolded version).
In these three algorithms, the first tree ($T_1$) is split by the keys in the second tree ($T_2$).  We
call $T_1$ the \emp{decomposed tree} and $T_2$ the
\emp{pivot tree}, denoted as $T_d$ and $T_p$ respectively.  The returned
tree of the algorithms is denoted as $T_r$.
We use $m=\min(|T_p|, |T_d|)$ and
$n=\max(|T_p|,|T_d|)$.
We denote the subtree rooted at $v \in T_p$ as $T_p(v)$, and the tree
of keys from $T_d$ that $v$ is operated with as $T_d(v)$ (i.e.,
\op$(v,T_d(v))$ or \opbase$(v,T_d(v))$ is called at some point in the algorithm.
This essentially means that $v$'s subtree in $T_p$ is processed with the tree $T_d(v)$ in a recursive call.
Note that $T_d(v)$ may not be a subtree in $T_d$, but is a tree of keys as a subset of $T_d$. We call such $T_d(v)$ a \emp{\subsettree{}} of the \decomposedtree{} $T_d$. We say $v\in T_p$ \emp{\dealswith} $T_d(v)$.
For $v \in T_p$, we refer to $|T_d(v)|$ as its \emp{splitting size}.

\begin{lemma}
The refold work can be asymptotically bounded by the unfold work.
\end{lemma}
\begin{proof}
  We note that during the process of tracking down the tree, we will
  refold the subtree if and only at least a subset of it was
  previously unfolded at some point in this algorithm. Since \refold{}
  costs $O(\B)$ work, it can be asymptotically bounded by the
  corresponding \unfold{} function invoked previously.
\end{proof}

\begin{theorem}
\label{thm:appjoinlarger}
For each function call to \op{} on
trees $T_p(v)$ and $T_d(v)$, the work done the \join{} (or \joinTwo{})
is asymptotically bounded by the work done by \tsplit{}.
\end{theorem}
\begin{proof}
In the following, we use $P$ and $D$ to denote
$T_p(v)$ and $T_d(v)$, respectively, for simplicity. Assume the return value is $R$. First of all, the work of \tsplit{} is $\Theta(\log |D|+\B)$.

For \intersection{} or \difference{}, the work of \join{} (or \joinTwo{})
is $O(\log|R|+\B)$.
Notice that \difference{} returns the keys in $D\backslash P$. 
Thus for both \intersection{} and \difference{} we have $R\subseteq D$. Therefore $|R|\le |D|$, which means the work done by \join{} or \joinTwo{} is no more than the work done by \tsplit{}.

For \union{}, first of all, we always call \join{} instead of \joinTwo{}. If $|P|\le |D|$, then $|R|\le 2|D|$.
\join{} costs work $O(\log |R|+B)=O(\log |P|+B)$, which is no more than $\Theta(\log |D|+\B)$.

Consider $|P|>|D|$. The subtrees $\lc(P)$ and $\rc(P)$,
which are used in the recursive calls, have size at least $\alpha|P|$ and at most $(1-\alpha)|P|$.
After combining with a subset of elements in $D$ (which has size smaller than $|P|$), the return value of each recursive call should have size at least $\alpha|P|$ and $(2-\alpha)|P|$. Denote these two trees from recursive calls as $t_l$ and $t_r$, respectively.
Note that $\alpha$ is a constant, so the difference of size between $t_l$ and $t_r$ is also a constant.
WLOG assume $|t_l|\ge |t_r|$. In the following, we discuss different cases of whether $t_l$ and $t_r$ are \composite{} or \simple{} trees.
We will show that, in all cases, joining $|t_l|$ and $|t_r|$ has work $O(\log |D|+\B)$.

\begin{enumerate}
  \item When both $t_l$ and $t_r$ are \simple{} trees. From \cref{lem:joincost}, \join{} costs $O(\B)$ work.
  \item When both $t_l$ and $t_r$ are \composite{} trees. From \cref{lem:joincost}, \join{} costs $O(\log \frac{|r_l|}{|t_r|})=O(1)$ work.
  \item When $t_l$ is a \composite{} tree, but $t_r$ is a \simple{} tree. This means that $|t_l|>\B$ and $|t_r|\le \B$.
  From \cref{lem:joincost}, \join{} costs $O\left(\B + \log\frac{|t_l|}{\B}\right)$ work. Note that, as stated above, $|t_r|\ge \alpha|P|$. Considering $|t_r|\le B$, we know that $|P|=O(B)$, which also indicates $|t_l|=O(\B)$. Plug this into the work of \join{} $O\left(\B + \log\frac{|t_l|}{\B}\right)$, we can get the bound $O(\B)$.
\end{enumerate}

In summary, in all cases, the work of \join{} or \joinTwo{} is asymptotically bounded by the corresponding \tsplit{} function.
\end{proof}

From the above two lemmas, we have that the total work of \op{} is
asymptotically bounded by the split work and unfold work. Next, we
prove the bounds for split work and unfold work, respectively.

\begin{lemma}
The unfold work is $O(\min(m\B, n))$.
\end{lemma}
\begin{proof}
  First, note that in our unfolded version of base cases, any \block{}
  needs to be unfold at most once. Each time the algorithm hits a
  \blockednode{s}, it unfold the entire subtree in $O(B)$ time.

  If $|T_p|=m$, we note that there are at most $m/\B$ \blockednode{s}
  in $T_p$ that needs to be unfolded, so the total work to unfold
  $T_p$ is $O(m)$. Each of the $O(m)$ entries in $T_p$ will cause an
  \unfold{} on at most one \block{} in $T_d$. Therefore, the total
  work to unfold $T_d$ is $O(m\B)$. On the other hand, note that each
  \block{} in $T_d$ can be unfolded at most once, which also means
  that the work of unfolding $T_d$ is $O(n)$. In summary, the work of
  all \unfold{} functions is $O(\min(m\B, n))$.

  If $|T_d|=m$ and $|T_p|=n$, there will be at most $O(m)$ nodes in
  $T_p$ used to \dealwith{} a \subsettree{} in $T_d$. Since $|T_d|=m$,
  the total work of unfolding $T_d$ is at most $O(m)$. Based on the
  same argument as above, the total work of unfolding $T_p$ is
  $O(m\B)$ because at most $O(m)$ \unfold{} functions are invoked, and
  is also $O(n)$ because there are at most $O(n)$ entries in $T_p$.
  The total work is also $O(\min(m\B, n))$.
\end{proof}

\begin{lemma}[Split work on \exploded{} trees]
\label{lem:appPAMsplit}
The total split work done on two \exploded{} weight-balanced trees of sizes $n$ and $m\le n$ is $O\left(m\log \frac{n}{m}\right)$.
\end{lemma}
This directly follows \cite{blelloch2016just}.

\begin{lemma}
The total \tsplit{} work is $O\left(m\log \frac{n}{m}+\min(m\B,n)\right)$.
\end{lemma}
\begin{proof}
  The total work for \tsplit{} functions can be viewed as two parts:
  the total work to done by \tsplit{} functions to traverse and split
  non-\blockednode{s}, and the work to expose and split the
  \blockednode{s}. Note that here ``non-\blockednode{s}''
  include both \internalnode{s} in \composite{} trees, and all the
  nodes in \exploded{} trees.

  We first note that in the base cases, the
  \tsplit{}$(\cdot,\cdot,\true)$ function must be working on a
  \exploded{} tree. As a result, the total work to traverse and split all non-\blockednode{s} can be asymptotically bounded by the split work considering if both $T_p$ and $T_d$ are fully expanded. This cost can be computed by \cref{lem:appPAMsplit}, which is $O\left(m\log \frac{n}{m}\right)$.

  We then consider all work done by the \tsplit{} functions on \blockednode{s} in the non-base cases. We will show it is $O(\min(mB,n))$.
  Note that in each \tsplit{}, this happens at most once, costing $O(B)$ work.
  If $|T_p|=m$, there can be $O(m/\B)$ such \tsplit{} function calls in the non-base cases, and thus the total non-base case split work on \blockednode{s} is $O(m)$.

  If $|T_p|=n$, we discuss in two cases. If $n/\B\le m$, there are $O(n/\B)$ \internalnode{s} in $T_p$, and thus there can be at most $O(n/\B)$ \tsplit{} calls. Therefore the total work in this case is $O(n)$, which is also $O(\min(mB,n))$. If $n/\B > m$, there are at most $O(m)$ such \tsplit{} function calls, since there are only $m$ nodes in $T_d$. In this case, the total work of this part is $O(m\B)$, which is also $O(\min(mB,n))$.
\end{proof}
%

\else
Note that the $O(m\B)$ term can be expensive when $m$ is large.
In fact, we can show a tighter bound using a more efficient (but more
complicated) base case. We show the bound in
\cref{thm:unionworkbetter}, and defer the algorithm and proof to the
full version. In our implementation, we use the version in
\cref{fig:mainalgo}, which has good performance in practice.

\begin{restatable}{theorem}{setefficient}
\label{thm:unionworkbetter}
There exist algorithms for \union{}, \intersection{} and \difference{}
  on two \ourtree{s} of sizes $m$ and $n$ ($n\ge m$) with work $\displaystyle O\left( m\log \frac{n}{m} + \min(n,m\B)\right)$ and span $O(\log n \log m)$.
\end{restatable}
\fi

\ifx\confversion\undefined
\subsection{Non-parallelizable Encoding Schemes}
\label{app:setspan}

As mentioned at the beginning of this section, if the encoding scheme
is not parallelizable, when we deal with a \blockednode{}, we have to
deal with it sequentially and this can affect the span bound of our
algorithms. Again we will use the set algorithms as examples.

\begin{theorem}
\label{thm:unionspan}
Consider the \union{} algorithm (and similar \intersection{} and \difference{} algorithms) in \cref{fig:mainalgo} on two \ourtree{s} of sizes $m$ and $n\ge m$ using encoding scheme $\mathcal{C}$.
When $\mathcal{C}$ takes $O(B)$ work and span to compress and decompress a block of size $B$, the span for these set algorithms is $O(B+\log n'(B+ \log m'/B))$, where $n'$ is the pivot tree size, and $m'$ is the decomposed tree size.
\end{theorem}

\begin{proof}
The algorithms need $O(\log |T_p|/B$ rounds to reach a \blockednode{},
where the \blockednode{} will be \exploded{}, taking $O(B)$ span.
Note that this $O(B)$ additional span is taken only once for each \blockednode{}, and they are all at the leaf level of $T_p$.
As a result, they do not add up and can be charged only once in the span.
Then the algorithm keeps
recursing until a \nil{} node is reached, which is $O(\log B)$ rounds. The total number of rounds of recursive calls is still $O(\log n')$.
In each of the recursive calls,
we need $O(B+\log m'/B)$ span to deal with \tsplit{} and \join{}. In total the span is $O(B+\log n'(B+ \log m'/B))$.
\end{proof}
\fi

\ifx\confversion\undefined
\section{Work Tradeoff between Updates and Queries}

In \ourtree{s}, we flatten and compress subtrees with size no more than $2B$ where $B$ is a predefined parameter.
This approach has several benefits: it is more space efficient, allows
for effective compression, and reduces the memory footprint for updates and queries.
The only disadvantage is that an insert or delete now has work $O(B+\log n)$, while the P-tree only uses $O(\log n)$ work.
We now show an alternative solution for \ourtree{s} which are updated
in-place such that the amortized work for an update is only $O(\log (n/B))$, at a cost of more work for queries (i.e., $O(B+\log n)$).

The basic idea is to leave each leaf node unsorted.
In addition, we can keep a linked list for all entries.
As such, for an update, we simply find the corresponding leaf node.
For an insertion, we add this entry to the end of the linked list of this leaf tree node.
For a deletion, we remove this node from the linked list.
For both cases, we update the counter for leaf node size, and split or merge if necessary.
This also works for inserting or deleting a batch of entries.
The only difference is that it is easier to mark tomb entries for batch deletions and physically delete the entries in the next leaf split or merge, since deleting multiple nodes in a linked list is hard (requiring list ranking), while marking tombs can be done when all entries are kept in an array.
Here we assume we can locate an entry using its unique identifier, so we can map the identifer to the position in the linked lists.
For instance, if the entries are vertices in a graph, then we can use vertex labels as the identifiers for the entries.

This approach increases the lookup cost to $O(B+\log n)$ since now we need to check all entries in a leaf node in the worst case (the work for \texttt{aug\_range} query is still $O(B+\log n)$ since we only need to check two leaf nodes).

An additional change is the leaf node size.
In the previous algorithm, leaf nodes have sizes between $B$ and $2B$.
In this setting, we relax it to be $B$ and $(2+3c)B$ for any constant $c>0$.
For instance, if $c=0.1$, then the leaf node size will be between $B$ and $2.3B$.
This change is needed to amortizing the split and merge work for the updates.
If $c=0$, we can have a tree with all leaf nodes containing $B$ entries.
If we delete any entry, the associated leaf nodes will contain $B-1$ entries, need to merge with a neighbor leaf and end up with having $2B-1$ nodes.
Then if we insert two entries in this leaf, we need to split again with $O(B)$ work.
To avoid this, we set a padded region of size $cB$ on both sides of the range---once resized, the new leaf node contains $(1+c)B$ to $(2+2c)B$ entries.
As such, we need to remove or insert another $cB$ records to trigger the next resizing, so the amortized work is $O(B/cB)=O(1)$ per update.

\begin{theorem}\label{thm:unsorted}
  A batch of $m$ insertions or deletions can be processed using $O(m\log(n/B))$ amortized work if the batch is unsorted, or $O(m\log(n/Bm))$ amortized work if the batch is sorted.
  The span is $O(\log (n/B)\log m+\log B)$ for both cases.
\end{theorem}

\begin{proof}
  We first assume the update batch is unsorted.
  We use the tree root to partition the batch, which takes $O(m)$ work and $O(\log m)$ span.
  Then we can recursively and in parallel update the left part of the batch and the left subtree, the right part of the batch and the right subtree.
  After they both finish, we join the two trees with the tree root.
  The base case is when the corresponding batch for the subtree is empty, or the tree goes to a leaf node.
  We terminate for the first case.
  For the second case, we update the leaf with work proportional to the update array size (concatenation for insertions, marking tombs for deletions).
  The update may trigger a clean-up for the leaf array (split if the array size is larger than $(2+3c)B$, merge if the size is smaller than $B$, or pack if the tombs occupy over a constant fraction of the array).
  Once a clean-up is triggered, the work is $O(B)$ and the span is $O(\log B)$, and as explained, the work is constant amortized to each previous update to this leaf node.

  Similar to the previous analysis, we can split the work into the split work, the join work, and the base case work.
  The base case work is constant per update.
  The split work is $O(m\log(n/B))$---each tree level will partition the update array and the total cost is $O(m)$ per tree level.
  Since the tree has size $n$ and $O(\log(n/B))$ levels, the total split work is $O(m\log(n/B))$.
  The join work is logarithmic in the subtree size, and at most $m$ leaf nodes are modified, so the overall work for join is $O(m\log(n/Bm))$, bounded by split work.
  Putting all pieces together, the work is $O(m\log(n/B))$ amortized.
  The span is $O(\log (n/B)\log m+\log B)$---$O(\log m)$ for split and join for $O(\log (n/B))$ levels, and $O(\log B)$ for the base case.

  If the update batch is sorted, the split becomes a binary search
  with cost logarithmic in the current update array size.
  In this case, the split work is bounded by the join work, so the total work becomes $O(m\log(n/Bm))$ amortized.
  The span remains unchanged.
\end{proof}

This new version of the \ourtree{} may of interest when updates are more frequent than queries, or the queries are more costly.
For instance, if the query is ``reporting the top-$k$ elements'', where $k$ can be large, the work for this query is $\Omega(k)$.
In this case, we can use this alternative version of the \ourtree{} to reduce the update cost.
Assume $k$ is static throughout the algorithm, we can set $B=k$ for the CPAM tree.
Based on \cref{thm:unsorted}, each update costs $O(\log(n/k))$ work.
For a query, we only need to look into the first leaf array which contains at least $k$ entries and at most $(2+3c)k$ entries.
Then we can use the classic algorithm~\cite{JaJa92} to find the $k$-th element from this array, and pack those no larger than it.
Hence, this query takes $O(k)$ work and $O(\log k\log \log k)$ span, and the span can be optimized to $O(\log k)$ using the deterministic sampling technique in~\cite{blelloch2010low}.
This is much more efficient than directly running the classic algorithm~\cite{JaJa92} with $O(n)$ work, and better than keeping the entire search tree (e.g., a P-tree) which has $O(\log n)$ work per update.

\fi

\hide{
\subsection{some lemmas}
\begin{lemma}
The expose work is $O(m)$.
\end{lemma}
\begin{proof}
  Note that all \expose{} functions in \op{} happens after checking if $T_2$ is not a \compressednode{}. Also, in the corresponding base case functions, \expose{} is called on $T_2$ only after $T_2$ is unfolded. Therefore, in both cases, the work for such an \expose{} function is a constant. Considering that there can be at most $O(m)$ recursive calls, the total expose work is $O(m)$.
\end{proof}

\begin{lemma}
In a \ourtree{} $T$, all \internalnode{s} form a weight-balanced tree with balance factor $\frac{\alpha}{2-\alpha}$.
\end{lemma}
\begin{proof}
  Consider an \internalnode{} $v$ in a \ourtree{} and its sibling $u$. Denote their parent as $p$. Assume the number of \internalnode{s} in a subtree of $T$ rooted at a node $x$ is $s(x)$. Recall that we use $w(u)$ and $w(v)$ to denote the weight of $u$ and $v$ in the \ourtree{}, respectively.
  Then we have:
  \begin{align*}
    \B(s(u)+1)\le&w(u)\le 2\B(s(u)+1)\\
    \B(s(v)+1)\le&w(v)\le 2\B(s(v)+1)
  \end{align*}
  Denote $\beta = \frac{1-\alpha}{\alpha}$. Considering the weight-balance invariant, we also have $\frac{1}{\beta} w(u)\le w(v)\le \beta w(u)$. If we look at the binary tree $T'$ formed by all the \internalnode{s} in this \ourtree{}, we denote the weight of a subtree in $T'$ rooted at a node $x$ as $w'(x)=s(x)+1$.
  By using the inequations we can get:
  \begin{align*}
  &\B(s(u)+1)\le w(u)\le \beta w(v)\le 2\B\beta(s(v)+1)\\
  \Rightarrow & w'(u) \ge 2\beta w'(v)\\
  \Rightarrow & w'(v) \le \frac{1}{2\beta+1}w'(p)\\
  \Rightarrow & w'(v) \le \frac{\alpha}{2-\alpha}w'(p)
  \end{align*}
  Since $v$ is either child of $u$, we can also get the symmetric conclusion $w'(v) \ge \frac{2-\alpha}{\alpha}w'(p)$. This shows that $T'$ is a weight-balanced tree with balancing factor $\alpha'=\frac{2-\alpha}{\alpha}$.
\end{proof}

\begin{lemma}
In a weight-balanced tree $T$, we can partition all tree nodes into $l$ layers, such that:
\begin{itemize}
  \item $l=O(\log |T|)$.
  \item For any two nodes $u$ and $v$ in the same layer, $u$ is not the parent of $v$.
  \item The number of nodes in layer $i$, denoted as $s_i$, is no more than $|T|/c^i$ for some constant $c$.
\end{itemize}
\end{lemma}
}

\section{Implementation}\label{sec:impl}
In this section, we describe \cpam{}, our implementation of
\ourtree{s}. 
\cpam{} is built in C++, based on the PAM framework~\cite{pam}.
Our
implementation of sequence and map primitives are mostly unchanged.
Most of the changes are to
introduce \blockednode{s},
to handle folding and unfolding in \join{},
to express the recursive functions using the
\expose{} primitive, and in some cases to add optimized base cases.

\myparagraph{Optimized Base Cases.}
We first implemented \union{} as in \cref{fig:mainalgo}, which
recursively calls \expose{} to access the left and right subtrees.
Although simple and theoretically efficient, in practice unfolding
\blockednode{s} into \exploded{} trees and recursing on these trees
requires additional memory allocations, and potentially more
cache-misses.
We therefore designed a new sequential base-case for \union{} when $|T_1| + |T_2| < \kappa$, where $\kappa$ is a
configurable \emph{base-case granularity}.
Our base-case works by writing both $T_L$ and $T_R$ into
a pre-allocated array $A$ of size $\kappa$ and merging them in-place
to perform the union. It then constructs a \ourtree{} from the result
in $A$.
Compared to the original version of \union{} that only
uses \expose{}, using the special base-case with $\kappa = 4B$ is 4.4x
faster, and using $\kappa = 8B$ is 6.7x faster ($B = 128$). We
observed similar improvements for some other commonly-used
primitives such as \filter{}, \mapreduce{}, \multiinsert{}, \multidelete{},
and \intersection{}.
We use $\kappa = 8B$ in our experiments.
\revised{We use a parallel granularity of $4B$, which is the threshold
for forking parallel tasks in algorithms such as \filter{}
and \union{}.}


\myparagraph{Persistence and Memory Management.}
\cpam{} uses a reference counting garbage collector for memory management.
\cpam{} provides functional ordered maps, and thus by default
does not modify the input trees. 
However, in certain cases an application may wish to modify a
tree in-place to save memory, e.g., when updates and queries are separated.
Although one could deal with in-place and functional updates separately,
this is not attractive.
\ifx\confversion\undefined
Instead, we designed a simple approach to handle both cases using the
same code.

%
Our approach is to store an additional bit indicating whether the
supplied node is visible solely to the current function, or whether
the node has some external observer, and should therefore be copied.
We refer to these special pointers to tree nodes with an additional
bit for visibility as \emph{\extraptr{s}}. When an \extraptr{} is
copied, e.g., an algorithm like \union{} wishes to use
it as part of the resulting tree, we copy this node if the visibility
bit is set or if the node has a reference count more than 1,
and otherwise we simply return the same node.
Similarly, when we \expose{} an \extraptr{} pointing to a
\regularnode{} $v$, we set the visibility bits on the children either
if the $v$'s visibility bit is set, or if $v$ has a reference count
greater than $1$. If $v$ was visible only to the caller, as an
optimization we return it as an additional result, allowing the caller
to potentially reuse this node.
Our approach lets us write simple algorithms which modify the tree
in-place when possible, and begin copying once it reaches subtree that
is visible to other observers.

\else
Instead, we designed a simple approach to handle both cases using the
same code, which we describe in Appendix~\ref{apx:experiments} due to
space constraints.
\fi


\myparagraph{Compression on Blocks.}
\cpam{} makes it easy to apply user-specified encoding
schemes.  Our data structure is templated over a type
representing a block encoding scheme (no encoding by default).
To add a new encoding scheme, users provide
a structure with methods that calculate the encoded size for a block,
encode the elements into a buffer, and decode elements from an encoded
buffer. This design allows users to specify encoding schemes based on
the underlying data type or application, such as text compression. For
example, it is easy to add new types of difference coding, e.g., using
$\gamma$-coding, which would obtain better space usage at the expense
of worse running time~\cite{shun2015smaller}.

%


\section{Applications}\label{sec:app}

In this section we describe four applications that we
implement using \cpam{}. Our inverted index, and range and interval
tree applications are based on the implementations from
PAM~\cite{pam}. Our graph processing application is based on
Aspen~\cite{dhulipala2019low}.
We focus on the key
features of the applications in the context of \ourtree{s} here.

\myparagraph{Inverted Index.}
\revised{We implement a weighted inverted index, similar to those used in
search engines. The inverted index maintains a top-level map
from words to document lists ($B=128$). Each document list is a map
from document id to an importance score ($B=128)$.}
The document lists are augmented to maintain the highest importance score.
The inverted index supports standard AND/OR
queries over words, returning results by rank, and top-$k$ (based on importance) queries.
The document ids are compressed using difference encoding, requiring less than two bytes
per document.

\myparagraph{2D Range Tree.}
\revised{The two-dimensional range tree is a top-level map from
$x$-coordinate to $y$-coordinate ($B=128)$. The tree is augmented so
that every internal node stores all $y$-coordinates in its subtree
(this is itself a set represented as a \ourtree{} with $B=16$).}
Updates can add and delete points, and queries can list of
or count the points in a given rectangular range.
The range tree supports count queries in $O(\log^2 n)$ time, which can
be batched to run in parallel.

\myparagraph{Interval Tree.}
The interval tree maintains intervals over the number line, for
example, representing the time of a TCP connection, or the time a user is
logged into some service. A stabbing query can report all or any
intervals that cross a given point. \revised{The intervals are
represented as an augmented tree from left-coordinate to
right-coordinate with $B=32$. The augmentation maintains the maximum
right-coordinate in the subtree. This allows stabbing queries in time
$O(k \log n)$ where $k$ is the number of intervals requested or
returned (whichever is less). Intervals can be inserted or deleted
in $O(\log n)$ time and can be batched to run in parallel.}

\hide{Difference encoding can be used for compression, and works
particularly well if the intervals tend to be short.}

\myparagraph{Graph Processing.}
\revised{Graphs are represented as a two-level structure similar to the
inverted index, with a top-level augmented tree (the \emph{vertex
tree}) from vertices to edge lists ($B=64$).
Each edge list is a map from neighbor-id to an edge-weight (or empty
when unweighted) called an \emph{edge tree} ($B=64$).}
The augmentation on the vertex tree maintains the
total number of edges in the graph.  We focus on unweighted graphs in
this paper but note that our implementation also supports weights.  As with
inverted indices, using difference encoding allows us to store an edge
using just 2--3 bytes on average including the bytes used for
\regularnode{s}.

On top of this representation, we implement graph algorithms
using the Ligra interface~\cite{ShunB2013}, including 
breadth-first search, maximal independent set, and single-source
betweenness centrality. Our implementations are based on the ones in
Aspen and GBBS~\cite{DhulipalaBS21, DhulipalaSTBS20}.
We design parallel batch-updates for our representation, which
are applicable in graph-streaming and batch-dynamic graph
algorithms.

\section{Experiments}\label{sec:exps}

\myparagraph{Experimental Setup.} We run experiments on a 72-core Dell
PowerEdge R930 (with two-way hyper-threading) with $4\times
2.4\mbox{GHz}$ Intel 18-core E7-8867 v4 Xeon processors (with a
4800MHz bus and 45MB L3 cache) and 1\mbox{TB} of main memory. Our
programs use a work-stealing scheduler for
parallelism~\cite{parlay20}.  We use 
\texttt{numactl -i all} to balance the memory allocations across the
sockets for parallel executions.
Unless otherwise mentioned, all of the reported numbers are run on 72
cores with hyper-threading.

\revised{
\myparagraph{Overview of Results}
We show the following experimental results in this section.

\begin{itemize}[topsep=0pt,itemsep=0pt,parsep=0pt,leftmargin=8pt]

\item \ourtree{}s are competitive with PAM for microbenchmarks
  (Section~\ref{subsec:microbench}) and applications including
  inverted indices (Section~\ref{subsec:inverted}) and 2D range
  queries and 1D interval queries
  (Section~\ref{subsec:intervalandrange}) while using 2.1x--7.8x less
  space.

\item Varying the block size $B$ for an \ourtree{} trades off off
  performance for space efficiency (Section~\ref{subsec:microbench}).
  For even a modest value of $B=128$, \ourtree{}s use only 1\% more
  space than a (static) compressed array.

\item For graph processing and streaming, \cpam{} uses 1.3--2.6x less
  space compared to \aspen{}, and is almost always faster than
  \aspen{} in all tested graph algorithms (Section~\ref{sec:graphs}).

\end{itemize}
}

\begin{figure*}

\footnotesize
\vspace{-1em}
\begin{minipage}[t]{1.15\columnwidth}
\begin{tabular}{@{}l@{ }c@{  }cr@{  }@{  }r@{  }@{  }rr@{  }@{  }r@{  }@{ }rr@{  }@{ }r@{ }@{  }r}
\toprule
\multirow{2}{*}{} & \multirow{2}{*}{$n$} & \multirow{2}{*}{$m$} & \multicolumn{3}{c}{\bf \ourtree{}} &  \multicolumn{3}{c}{\bf \ourtree{} (Diff)} & \multicolumn{3}{c}{\bf \ptree{} (PAM)} \\

& & & $T_{1}$ & $T_{144}$ & {Spd.} & $T_{1}$ & $T_{144}$ & {Spd.} & $T_{1}$ & $T_{144}$ & {Spd.} \\
\midrule
\multicolumn{12}{@{}l}{\textbf{No augmentation}} \\
\midrule
  {\emph{Size (GB)}}     & $10^{8}$ & ---        & 1.61  & ---     & ---    & \best{0.926}& ---    & ---         & 4.00  & ---    & --- \\
  {\emph{Build}}         & $10^{8}$ & ---            & 5.55  & 0.186 & 29.8     & 5.71 & \best{0.180} & 31.7         & 5.94  & 0.221  & 26.8 \\
  {\emph{Union}}         & $10^{8}$ & $10^{8}$       & 5.33  & \best{0.088} & 60.5     & 6.29 & 0.089  & 70.6        & 8.97  & 0.168  & 53.3 \\
  {\emph{Union}}         & $10^{8}$ & $10^{5}$       & 1.09  & 0.021 & 51.9     & 1.28 & 0.022 & 58.1         & 0.206  & \best{0.0038} & 54.2 \\
  {\emph{Intersect}}     & $10^{8}$ & $10^{8}$   & 4.35  & \best{0.065} & 66.9     & 5.68 & 0.081 & 70.1         & 9.50  & 0.139  & 68.3  \\
  {\emph{Difference}}    & $10^{8}$ & $10^{8}$  & 3.00  & \best{0.055} & 54.4     & 3.55 & 0.056 & 63.3         & 8.17  & 0.123  & 66.4 \\
  {\emph{Map}}           & $10^{8}$ & $10^{8}$      & 0.859 & 0.037 & 22.9     & 1.14 & \best{0.023} & 49.5        & 1.32 & 0.091 & 14.5 \\
  {\emph{Reduce}}        & $10^{8}$ & ---     & 0.306 & 0.018 & 17.0     & 0.308 & \best{0.0092} & 33.4       & 1.60 & 0.034 & 47.0 \\
  {\emph{Filter}}        & $10^{8}$ & ---      & 0.997 & 0.028 & 35.6     & 1.24 & \best{0.018} & 68.8         & 1.90  & 0.0524 & 36.2 \\
  {\emph{Find}}          & $10^{8}$ & $10^{8}$       & 103   & 1.17 & 88.0      & 125 & 1.23 & 101.6          & 105.5 & \best{1.05}   & 100.4 \\
  {\emph{Insert}}        & $10^{8}$ & $10^{6}$       & 0.829 & --- & ---        & 1.42 & ---  & ---           & \best{0.773} & --- & --- \\
  {\emph{Multi-Insert}}  & $10^{8}$ & $10^{8}$ & 18.8  & 0.332 & 56.6     & 19.9 & \best{0.323} & 61.6        & 9.67  & 0.338  & 28.6  \\
  {\emph{Range}}         & $10^{8}$ & $10^{6}$       & 11.5  & 0.318 & 36.1     & 13.1 & 0.226 & 57.9         & 3.77  & \best{0.0738} & 45.6 \\  
\midrule
\multicolumn{12}{@{}l}{\textbf{With augmentation}}\\
\midrule
  {\emph{Size (GB)}} & $10^{8}$ & ---            & 1.63 & ---        & ---    & \best{0.936} & ---  & ---        & 4.80 & --- & --- \\
  {\emph{Build}}   & $10^{8}$ & ---              & 5.66 & 0.197 & 28.7   & 5.84 & \best{0.186} & 31.3       & 6.48  & 0.246 & 26.3 \\
  {\emph{Union}}   & $10^{8}$ & $10^{8}$         & 5.52 & 0.098 & 56.3   & 6.52 & \best{0.090} & 72.4       & 10.13 & 0.196 & 51.6 \\
  {\emph{AugRange}}   & $10^{8}$ & $10^{7}$      & 12.3 & 0.331 & 37.1   & 13.9 & 0.234 & 59.4       & 4.80  & \best{0.082} & 58.5 \\
  {\emph{AugFilter}}   & $10^{8}$ & ---     & 12.2 & 0.333 & 36.6   & 13.6 & 0.234 & 58.1       & 4.95  & \best{0.081} & 61.1 \\
\end{tabular}
\captionof{table}{\small \textbf{Microbenchmark results.} We fix $B=128$ for \ourtree{s}.
$n$ is the tree size. For set functions and \texttt{multi-insert}, $m\le n$ is the size of the other set (batch). For other functions,
$m$ is the number of queries tested.
$T_1$ is the sequential running time. $T_{144}$ is parallel running time using 72 cores (144 hyperthreads). \emph{Diff} means difference encoding.
We highlight the best parallel running time (or size) per experiment in green and underlined.
\label{table:microbench}
}
\end{minipage}\hfill
\begin{minipage}{.85\columnwidth}
  \hspace{-.1in}
  \includegraphics[width=\columnwidth]{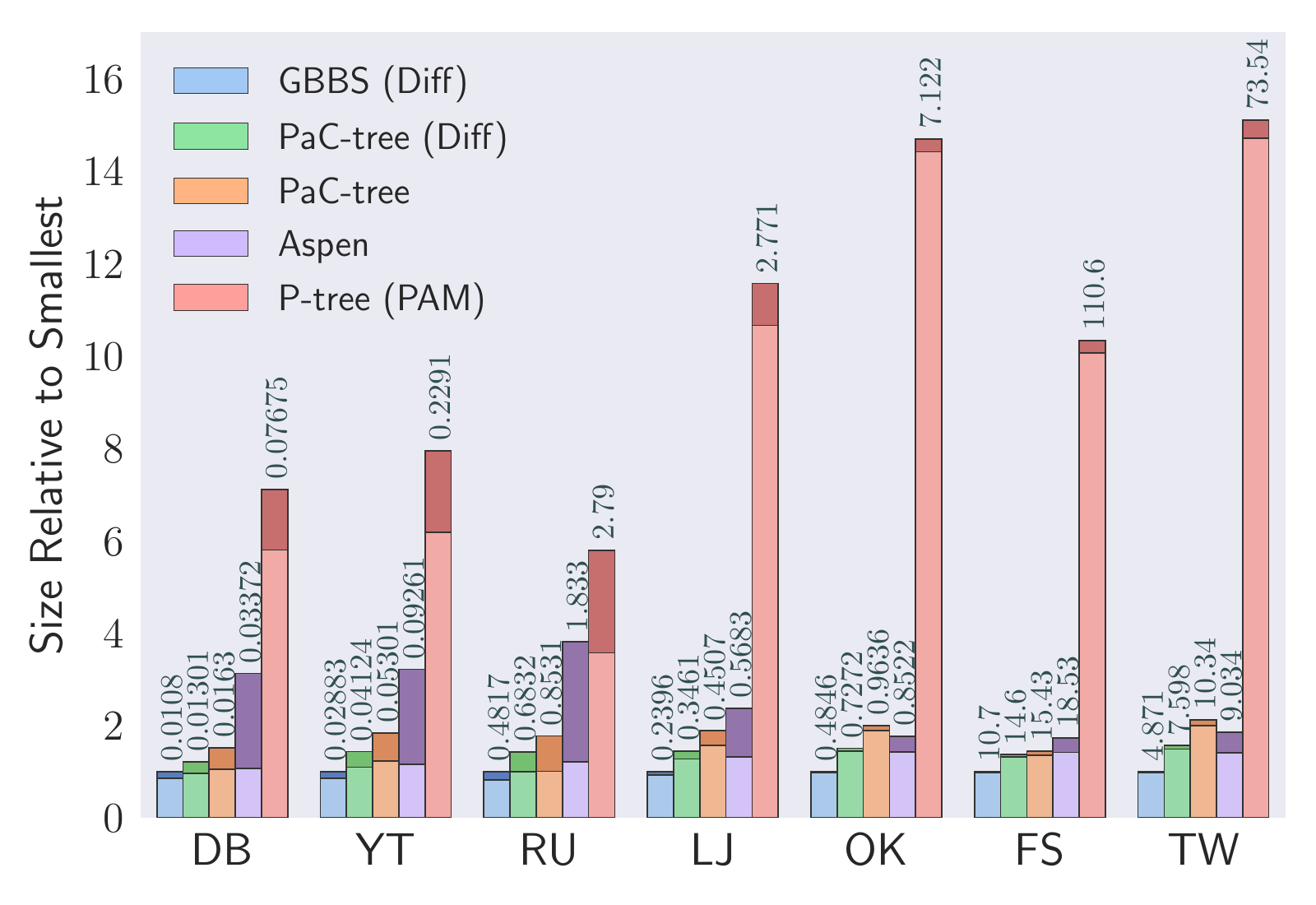}
    \caption{\small
    \textbf{Relative space usage of different graph representations.}
    GBBS (Diff) is our \emph{static} baseline compressed graph representation.
    \ourtree{} uses \ourtree{}s for vertex and edge trees, and
    \ourtree{} (Diff) difference encodes both trees.
    Aspen uses \ptree{}s for the vertex tree and \ctree{}s with
    difference encoding for edge trees.
    \ptree{} (PAM) uses \ptree{}s for the vertex and edge trees.
    The values on top of each bar are the memory usage in GiB.
    \label{fig:relative_graph_sizes}}
\end{minipage}\hfill
\end{figure*} 

\subsection{PaC-Tree Performance}\label{subsec:microbench}
We begin by studying the performance and space of \ourtree{s} on a set
of microbenchmarks and compare with \ptree{s} from \pam{}.  All
experiments in this section use maps and augmented maps where the keys
and values are both 64-bit integers. Unless otherwise mentioned
\ourtree{s} use $B=128$.

\myparagraph{Microbenchmark Performance.}
\cref{table:microbench} shows the results on \ourtree{s},
\ourtree{s} with difference-encoding (DE), and \ptree{s} for a representative
subset of the map and sequence primitives.
The speedups for both types
of \ourtree{s} range from 28.7--101x and are largest for the version
using DE due to additional work for difference encoding.
In absolute
running time, \ourtree{s} with DE are usually slower than \ourtree{s}
due to compression and decompression costs, but the overhead is mostly within 10\%.

In most of the primitives tested, \ourtree{s} are faster than \ptree{s}
while also using 2.5x less space.
For example, \ourtree{s} are 1.68x faster than \ptree{s} in \union{} on
two trees of sizes $10^8$.
We note that in this case, the union processes the entirety of both
input trees, and so the more cache-friendly processing of blocks in
\ourtree{s} results in lower time. However, if sizes of the two trees
are different, the work for \union{} only depends on the
smaller size.  In this case, since the cost of \union{}
using \ourtree{s} has an additional $O(mB)$ term compared with
\ptree{s}, \ourtree{s} are 5.5x slower than \ptree{s}.  However, we expect
better performance for smaller block sizes ($B < 128$), which we
discuss next.

\hide{
\begin{figure}[!t]
\vspace{-1em}
  \includegraphics[width=0.4\textwidth]{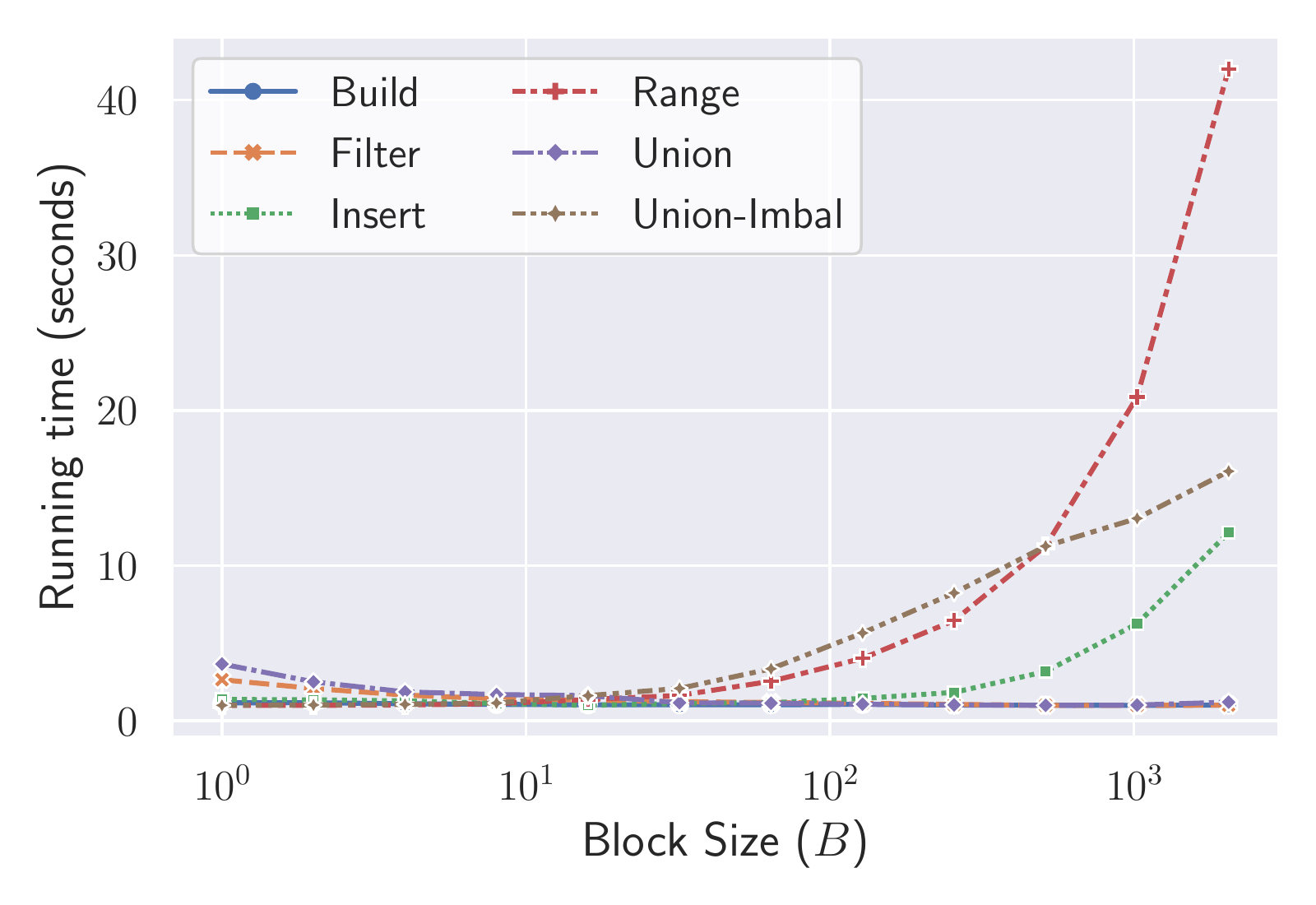}
  \vspace{-.5em}
    \caption{
  \small
  Primitive running times for \ourtree{} storing $10^8$ key-value
  pairs (8 bytes each) vs. block size $B$. $\mathsf{Union}$,
  $\mathsf{Intersection}$ and $\mathsf{Difference}$ all work on two trees with $10^8$ elements.
  $\mathsf{Union}$-$\mathsf{Imbal}$ takes the union of trees with
  $10^8$ and $10^5$ elements.}\label{fig:primitivetime_vs_blocksize}
  \vspace{-1em}
\end{figure}
}

\myparagraph{Effect of Varying $B$ on Performance.}
\cref{fig:primitivetime_vs_blocksize} shows the results of
varying the block size $B$, on the performance of various operations.
Most operations obtain speedups as $B$ is increased up until $B=16$.
For the sequential operations, such as \find{} and \range{}, we see a
steady increase in the running time for $B > 16$ and see a similar
trend for $\mathsf{Union}$-$\mathsf{Imbal}$, which takes the union of
trees with $10^8$ and $10^5$ elements. This slowdown with increasing
$B$ is due to the extra $O(mB)$ term in the work of
\union{}.  For the smallest block
size ($B=1$), our running time matches that of \ptree{s} on this
operation.

\hide{
\begin{figure}[!t]
  \includegraphics[width=0.4\textwidth]{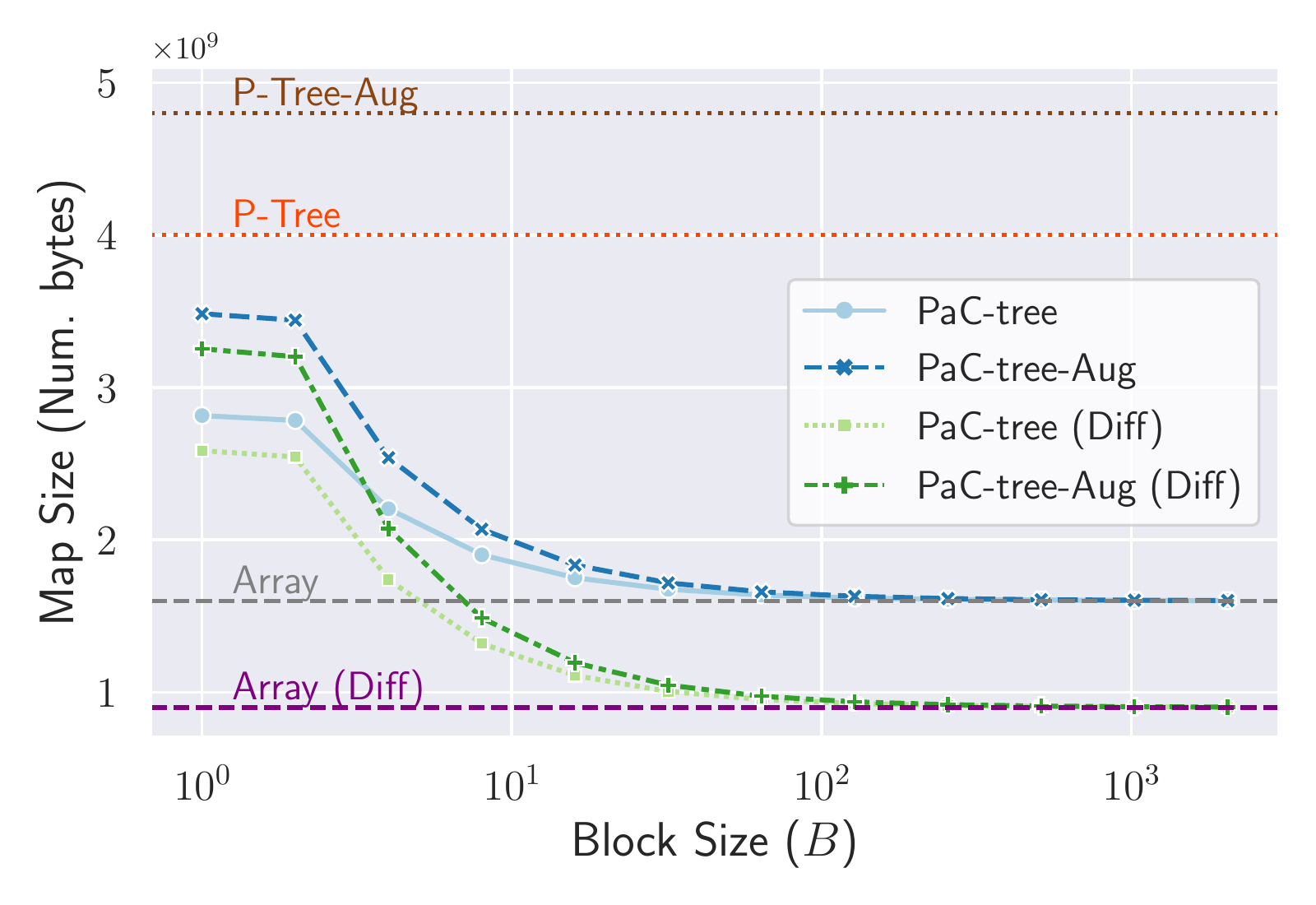}
  \caption{
  \small
  Size in bytes of \ourtree{} and \ourtree{} with difference encoding
  storing $10^8$ key-value pairs (8 bytes each) as a function of the
  block size, $B$. For augmented maps (-Aug) the augmented values are
  8 bytes each. The grey line shows the number of bytes to store the
  $10^8$ elements in an array and the purple line shows the bytes used
  to store the difference encoded keys in a single array using byte
  encoding, leaving values unmodified (8 bytes
  each).}\label{fig:size_vs_blocksize}
\end{figure}
}

\myparagraph{Space Usage.}
For $B=128$, \ourtree{s} obtain a 2.48x reduction in space usage
compared to using \ptree{s}, and a further 1.73x reduction in space
usage by using difference encoding.
The $10^{8}$ pairs stored in the experiments require 1.6GB of memory
to represent as a single flat array, which is also a lower bound for
the space usage of a search tree structure. To understand how close
\ourtree{s} come to this lower bound, we study the space usage of
unaugmented maps using \ourtree{s} as a function of the block size $B$
(\cref{fig:size_vs_blocksize}).
%
Using $B=32$, \ourtree{s} are only 1.05x larger than the lower bound
and using $B=128$, it is just 1.01x larger than the lower bound.
For $B=128$, just 1.1\% of the allocated memory is used for
\regularnode{s} and metadata in the \blockednode{s}.  
These savings are obtained without using any additional encoding.
%
Applying difference encoding improves the space by 1.77x over the
unencoded trees and the array lower bound, and is only 1.03x larger
than the space used to difference encode all of the keys in a single
array, leaving the values uncompressed, which is a lower bound for
a search tree structure using difference encoding for such input.

Using \ourtree{s} requires much lower space overhead for augmentation
compared to \ptree{s} (\cref{fig:size_vs_blocksize}).
For \ptree{s}, adding 8 byte augmented values increases the size of
the maps by 20\%, whereas \ourtree{s} (both with and without
difference encoding) using $B=128$ incurs only a 1\% increase in space
for the augmented values.  The savings comes from only storing a
single augmented value per \blockednode{}, which only uses extra space
proportional to $n/B$ augmented values.

\ifx\confversion\undefined
\subsection{Comparison with Collections in Spark}\label{sec:spark}
We compared \cpam{} with the shared-memory parallel implementation of
Apache Spark on a simple benchmark drawn from the Apache Spark
tutorial. The benchmark first loads the same Wikipedia corpus that we
use for our inverted index application (1,943,575,146 words in
8,125,326 documents). The first example then tokenizes the dataset
into words, and computes the longest word length. The second example
computes the most frequently occuring word by using the
\texttt{reduceByKey} primitive in Spark to group common words and
compute the mostly frequently occuring word using a reduce. We use
in-memory caching for the intermediate mapped dataset and report the
fastest (cached) time.

For the first example, Spark takes 46.9 seconds for the first
(uncached) run, and the subsequent average (cached) time is 21.5
seconds. Our \cpam{} implementation, where the dominant cost is the
memory-bound parsing step, requires 6.57 seconds on average (3.2x
faster than the cached time). On the second example, Spark takes 96.3
seconds for the first run, and 72.5 seconds for the average cached
time. For the second example, the dominant cost for \cpam{} is the
parallel sort (we use a parallel sample-sort). The end-to-end time is
14.6 seconds, which is 4.9x faster than the Spark cached time.

We also tried evaluating the same set of sequence benchmarks shown
in \cref{fig:comp} using Spark, but observed significantly worse
running times for all of the sequence primitives (up to 2 orders of
magnitude worse performance; e.g., reduce on a 100M element sequence
with 8-byte elements takes 2.07 seconds, whereas \cpam{} takes 0.00865
seconds). The slowdown could be due to fixed parallelization overheads
in Spark, although their word counting example which we studied above
performs reasonably well despite working over a significantly larger
dataset (nearly 2 billion words).
\fi

\subsection{Inverted Index}\label{subsec:inverted}
Next, we study our performance on the inverted index application. We
run the application on documents derived from a large Wikipedia
dataset also used by PAM for a fair comparison.
The dataset is processed by removing all markup, converting characters
that are not alphanumeric  to whitespace and making all words case
insensitive~\cite{pam}. The processed dataset contains 1.94 billion
words over 8.13 million documents. Like PAM, our evaluation measures
the performance of (1) building an index over (words, doc\_id, weight)
triples and (2) running queries that fetch the posting lists for two
words, compute the intersection of the lists, and select the top 10
documents by weight.

Table~\ref{table:tree_applications} shows the results of the
experiment. For building the index, our implementation achieves 76x
speedup and our parallel running times are comparable with those of
PAM (at most 1.1x slower). For the queries, we observe that the
unencoded trees achieve essentially the same parallel time as PAM,
whereas the difference encoded trees are 1.18x slower due to the
higher cost of intersection operations in our difference encoded
implementation. The space usage using \ourtree{s} is much smaller than
that of PAM, being 3.84x smaller without encoding and 7.81x smaller
using a custom encoder that combines difference encoding for the keys
with byte-encoding for the integer values (weights).



\begin{table}[!t]
\footnotesize
\centering
\begin{tabular}[!t]{ll@{  }@{  }@{  }c@{  }@{  }@{  }cccrrr}
\toprule
& Library & Space & Method & $n$ & $m$ & $T_{1}$ & $T_{144}$ & Spd. \\

\midrule
\parbox[t]{2.5mm}{\multirow{6}{*}{\rotatebox[origin=c]{90}{\bf Inverted Index}}} &
\multirow{2}{*}{\ourtree{}} & \multirow{2}{*}{8.29}
  & {\emph{Build}}   & $10^{8}$ & ---            & 746 & 9.73 & 76.6  \\
  & & & {\emph{Query}}   & $10^{8}$ & $10^{8}$       & 341 & \best{4.46} & 76.4 \\

\cmidrule(lr){2-9}

& \multirow{2}{*}{\ourtree{} (D)} & \multirow{2}{*}{\best{4.07}}
  & {\emph{Build}}   & $10^{8}$ & ---            & 754 & 9.81 & 76.8  \\
& & & {\emph{Query}}   & $10^{8}$ & $10^{8}$       & 367 & 5.32 & 68.9 \\
\cmidrule(lr){2-9}
  & \multirow{2}{*}{\ptree{} (PAM)} & \multirow{2}{*}{31.9} &{\emph{Build}}   & $10^{8}$ & ---  & 575 & \best{8.86} & 64.9 \\
&                     & & {\emph{Query}}   & $10^{8}$ & $10^{8}$                    & 313 & 4.48 & 69.8 \\
\midrule

\parbox[t]{2.5mm}{\multirow{4}{*}{\rotatebox[origin=c]{90}{\bf Interval}}} &
\multirow{2}{*}{\ourtree{}} & \multirow{2}{*}{\best{0.812}} &
{\emph{Build}}       & $10^{8}$ & ---            & 10.9 & \best{0.179} & 60.8 \\
& & & {\emph{Query}}   & $10^{8}$ & $10^{8}$       & 60.8 & \best{0.525} & 115.8 \\
\cmidrule(lr){2-9}
& \multirow{2}{*}{\ptree{} (PAM)} & \multirow{2}{*}{3.54} &
 {\emph{Build}}       & $10^{8}$ & ---            & 11.6 & 0.271 & 42.8  \\
& & & {\emph{Query}}   & $10^{8}$ & $10^{8}$       & 54.3 & 0.628 & 86.4 \\
\midrule

\parbox[t]{2.5mm}{\multirow{6}{*}{\rotatebox[origin=c]{90}{\bf Range}}} &
\multirow{3}{*}{\ourtree{}} & \multirow{3}{*}{\best{40.3}} &
{\emph{Build}}       & $10^{8}$ & ---            & 164  & \best{2.71} & 60.7 \\
& & & {\emph{Q-Sum}}   & $10^{8}$ & $10^{6}$   & 54.2 & \best{0.629} & 86.1 \\
& & & {\emph{Q-All}}   & $10^{8}$ & $10^{3}$   & 7.20 & \best{0.266} & 27.0 \\

\cmidrule(lr){2-9}
& \multirow{3}{*}{\ptree{} (PAM)} & \multirow{3}{*}{89.6} &
{\emph{Build}}       & $10^{8}$ & ---            & 169 & 2.84 & 59.6 \\
& & & {\emph{Q-Sum}}   & $10^{8}$ & $10^{6}$   & 60.7 & 0.735 & 82.5 \\
& & & {\emph{Q-All}}   & $10^{8}$ & $10^{3}$   &  21.6 & 0.552 & 39.1

\end{tabular}
\captionof{table}{\small \textbf{Build and query times and space usage in GiB
  for inverted index, interval tree, and range tree applications.} $T_1$ is the single-thread time,
  $T_{144}$ is the 72-core time using hyper-threading, and Spd. is the
  parallel speedup. The best parallel running time (or
  size) is highlighted in green and underlined per experiment.}
\label{table:tree_applications}
\vspace{-2em}
\end{table}

\subsection{Interval and Two-Dimensional Range Trees}\label{subsec:intervalandrange}
We benchmark our interval and two-dimensional range trees as in
PAM~\cite{sun2019parallel}.
%
We build our interval tree on $10^8$ intervals, and for queries run
stabbing queries over $10^8$ points in parallel. We observe that both
building and querying the trees achieves good parallel speedup
(60--115x).
\ourtree{s} are 1.51x faster than PAM in construction, and is 1.19x
faster for queries.
Overall we find that \ourtree{s} enable better
performance than PAM while using 4.37x less space.

We build our range trees on $10^8$ uniformly random points in the
plane between $(0, 0)$ and $(1e8, 1e8)$. We run two types of queries:
the first count the number of points in the range (Q-Sum), and the
second returns all points in the range. We tuned the window sizes used
in our queries to match the settings evaluated by PAM (around $10^6$ points
returned per query). Both \ourtree{s} and \ptree{s} build the data
structure in a similar amount of time. \ourtree{s} achieve better
performance than \ptree{s} for both queries, being 1.16x faster for
Q-Sum and 1.96x faster for Q-All queries, likely due to requiring
fewer cache-misses when processing the tree to output the points
within a given range. The range tree application using PAM has
previously been compared with range trees in CGAL~\cite{cgal-range}
and was shown to outperform it~\cite{sun2019parallel}.

For space usage, \ourtree{s} result in 2.18x less space compared to
PAM.
We note that
95\% of the space used in PAM is for the \ptree{s} stored as augmented
values in each node (representing the union of the $y$-coordinates in
the subtree). The majority of our savings come from compressing the
augmented trees using \ourtree{s} which results in a 2.53x less space
for the inner trees, and 2.18x less space overall.


\begin{figure*}[!t]
\vspace{-1em}
\hspace{-1em}
\begin{minipage}{.71\columnwidth}
    \includegraphics[width=\columnwidth]{figures/primitivetime_vs_blocksize}
  \vspace{-.5em}
\end{minipage}
\hspace{-1em}
\begin{minipage}{0.71\columnwidth}
  \includegraphics[width=\columnwidth]{figures/size_vs_blocksize}
  \vspace{-.5em}
\end{minipage}
\begin{minipage}{.71\columnwidth}
  \includegraphics[width=\columnwidth]{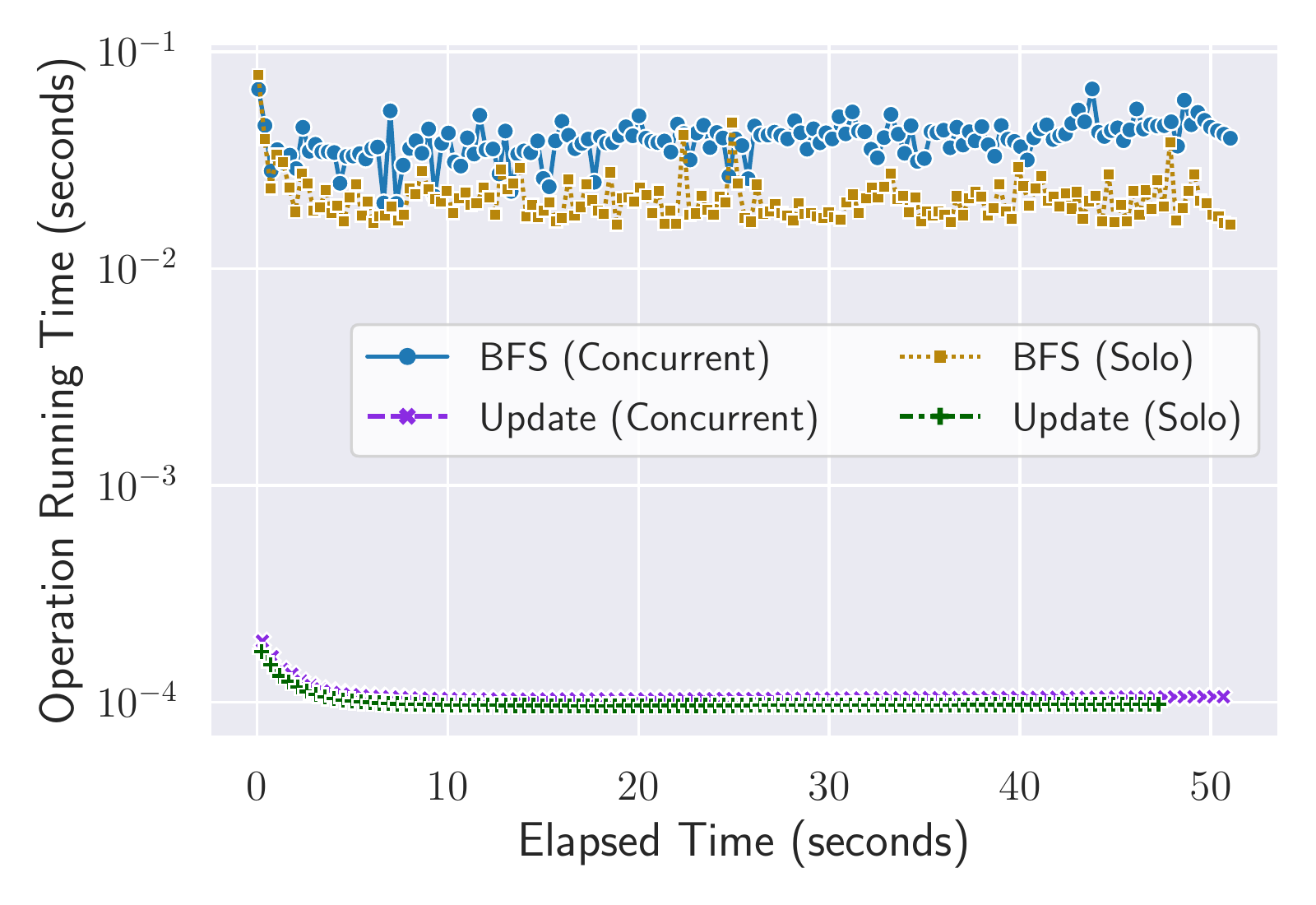}
\end{minipage}\\
\begin{minipage}[t]{.6\columnwidth}
    \caption{
  \small
  \textbf{Primitive running times for \ourtree{s} vs. block size $B$}.
  We use $10^8$ key-value pairs (8 bytes each). $\mathsf{Union}$,
  $\mathsf{Intersection}$ and $\mathsf{Difference}$ all work on two trees with $10^8$ elements.
  $\mathsf{Union}$-$\mathsf{Imbal}$ takes the union of trees with
  $10^8$ and $10^5$ elements.\label{fig:primitivetime_vs_blocksize}}
\end{minipage}\hfill
\begin{minipage}[t]{.75\columnwidth}
  \caption{
  \small
  \textbf{Size of \ourtree{s} (with or without DE) as a function of block size $B$.}
  We use $10^8$ key-value pairs (8 bytes each).
  For augmented maps (-Aug), augmented values are
  8 bytes each. The grey line shows the number of bytes to store the
  $10^8$ elements in an array and the purple line shows the bytes used
  to store the difference encoded keys in a single array using byte
  encoding. 
  \label{fig:size_vs_blocksize}}
\end{minipage}\hfill
\begin{minipage}[t]{.65\columnwidth}
    \caption{\small \textbf{Performance of concurrent updates and queries.}
    The time series plot illustrates running
    times when running BFS queries with batch-insertions
    of edges concurrently (Concurrent), and when queries and updates
    are run individually (Solo) on the LiveJournal graph.\label{fig:concurrent_update_query}}
\end{minipage}
\end{figure*}
\hide{
Relative space of different choices of map datatype in our graph
    representation. The acronyms indicate the type of tree used for
    the vertex trees (-V) and edge trees (-E) where $\mathsf{P}$ is
    \ptree{} and $\mathsf{C}$ is \ourtree{}. For example
    $\mathsf{PC}$ means that the vertex trees use \ptree{s} and the
    edge trees use \ourtree{s}. The best solution is
    $\mathsf{\ourgraph{}}$, which are nested \ourtree{s} using
    difference encoding. on top of each bar are the size of
    the representation in GiB.
} 
\begin{table}
\setlength{\tabcolsep}{4pt}
\footnotesize
\centering
\begin{tabular}[!t]{lrr rr c}
\toprule
\textbf{Graph} & \textbf{Vertices} & \textbf{Edges} & \textbf{Ours} & \textbf{Aspen} & $\frac{\text{Aspen}}{\text{Ours}}$ \\ 
\hline
{\emph{ DBLP  (DB) }  }           & 425,957          &2,099,732       & 0.0130   & 0.03409 & 2.62x \\
{\emph{ YouTube  (YT) }  }        & 1,138,499        &5,980,886       & 0.0412  & 0.0934 & 2.26x \\
{\emph{ USA-Road  (RU) }  }       & 23,947,348       &57,708,624      & 0.683   & 1.843 & 2.69x \\
{\emph{ LiveJournal (LJ) }  }    & 4,847,571        &85,702,474      & 0.346  & 0.527  & 1.52x \\
{\emph{ com-Orkut   (CO) }  }    & 3,072,627        &234,370,166     & 0.727  & 0.893 & 1.22x \\
{\emph{ Twitter    (TW)  }  }    & 41,652,231       &2,405,026,092   & 7.59  & 9.42 & 1.23x \\
{\emph{ Friendster   (FS)   }  } & 65,608,366       &3,612,134,270   & 14.6  & 19.1 & 1.30x \\
\end{tabular}

\captionof{table}{\small
\textbf{Statistics about tested graphs and memory usage of \ourtree{} and \aspen{} in GiB.}
\hide{
  Statistics about our input graphs,
  memory usage using nested \ourtree{s} using difference encoding (in GiB),
  memory usage of the Aspen graph representation (in GiB), and
  $\frac{\text{Aspen}}{\text{CC(D)}}$, the multiplicative factor
  reduction in space from using \ourtree{s}.}
\hide{
\captionof{table}{\small Statistics about our input graphs,
  memory usage using nested \ourtree{s} using difference encoding
  (CC(D)) in GiB, memory usage of the Aspen graph representation in
  GiB, and $\frac{\text{Aspen}}{\text{Ours}}$, the multiplicative
  factor reduction in space from using \ourtree{s}.
}}
\label{table:sizes-and-memory}}
\vspace{-1em}
\end{table}

\hide{
\begin{figure}[!t]
  \includegraphics[width=0.4\textwidth]{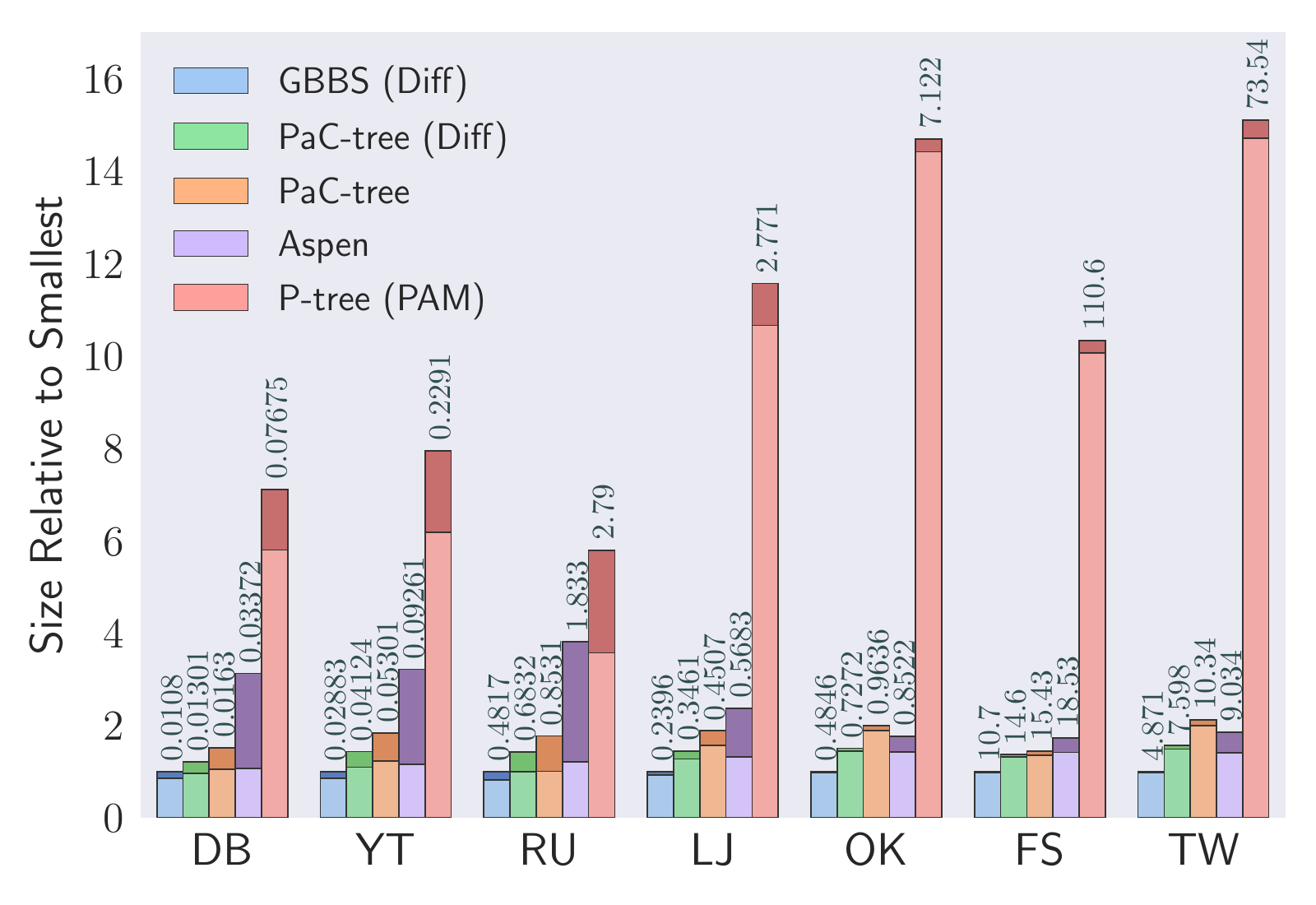}
    \caption{\small

    Relative space of different choices of map datatype in our graph
    representation. The acronyms indicate the type of tree used for
    the vertex trees (-V) and edge trees (-E) where $\mathsf{P}$ is
    \ptree{} and $\mathsf{C}$ is \ourtree{}. For example
    $\mathsf{PC}$ means that the vertex trees use \ptree{s} and the
    edge trees use \ourtree{s}. The best solution is
    $\mathsf{\ourgraph{}}$, which are nested \ourtree{s} using
    difference encoding. The values on top of each bar are the size of
    the representation in GiB. 
    }\label{fig:relative_graph_sizes}
\end{figure}
}

\subsection{Graph Processing and Graph Streaming}
\label{sec:graphs}
Our last set of experiments study the performance of \ourtree{s} for a
set of standard benchmarks from the graph processing and graph
streaming literature. Our evaluation roughly follows Aspen's and we
compare our performance and space usage with that of Aspen and its
\ctree{} implementation.

\ifx\confversion\undefined
\myparagraph{Graph Data.}
\emph{DBLP} is co-authorship network based on research papers in
computer science. \emph{YouTube} is a social network graph based on
YouTube.
\emph{USA-Road (RO)} is an undirected road network from the DIMACS
challenge~\cite{dimacs-shortest-paths}.  \emph{LiveJournal (LJ)} is a
directed graph of the LiveJournal social
network~\cite{boldi2004webgraph}. \emph{com-Orkut (CO)} is an
undirected graph of the Orkut social network. \emph{Twitter (TW)} is a
directed graph of the Twitter network~\cite{kwak2010twitter}.
\emph{Friendster (FR)} is an undirected graph describing friendships
from a gaming network. The DBLP, YouTube, and Friendster graphs are
obtained from the SNAP dataset~\cite{leskovec2014snap}.
We note that some of our inputs (like the LiveJournal graph) are
originally directed, and we symmetrize them before applying our
algorithms to maintain consistency with prior work on
Aspen~\cite{dhulipala2019low} and GBBS~\cite{DhulipalaBS21} that
symmetrize graphs in their evaluations. \cref{table:sizes-and-memory}
shows information about our graph inputs, including the number of
vertices, edges, and space used.
\else
\myparagraph{Graph Data and Space Usage.}
Most of the graphs we study are Web graphs and social
networks which are low-diameter graphs that are frequently used in
practice. To also test on high-diameter graphs, we ran our
implementations on a road network.  Complete details about our inputs
are in the full version of the paper.  \cref{table:sizes-and-memory}
shows information about our graph inputs, including the number of
vertices, edges, and space used.
\fi

\hide{ 
We evaluate four graph representations based on whether vertex and
edge trees use \ourtree{s} or \ptree{s}, and whether difference encoding
is used. \cref{fig:relative_graph_sizes} shows the relative size of
each graph format. We see that the smallest format in all cases is
\defn{\ourgraph{}}, which applies \ourtree{s} with difference encoding
for both vertex and edge trees. Using this format yields a space
improvement of between 4--9.7x over using \ptree{s} for both trees. For
the graphs with high average-degree, most of the savings come from
using \ourtree{s} for the edge trees. However, for the graphs with
lower average degree, like YouTube and USA-Road, which have average
degrees of just 5.29 and 2.41 respectively, we obtain an additional
1.62--1.83x space improvement by using \ourtree{s} for the vertex
trees. Adding difference encoding to both trees yields between
1.05--1.32x space improvement.

Next, we compare our space usage with Aspen's.
\cref{table:sizes-and-memory} shows the space used by the most
memory-efficient implementation using \ourtree{s} (\ourgraph{}), and
Aspen using \ctree{s}.
Note that \ctree{s} in Aspen are also
difference encoded, so the main difference between the two
representations is that \ourgraph{} also uses \ourtree{s} to chunk the
vertex tree, and that \ourtree{s} employ a deterministic strategy for
chunking. \ourgraph{} achieves consistently lower space compared with
Aspen, ranging between 1.3x for Friendster, our largest graph, to a
maximum space improvement for 2.69x on USA-Road, our sparsest graph.
The space savings come from chunking the vertex trees which is not
possible in Aspen since the \ctree{} implementation is specialized for
edge trees.
Adding difference encoding to both trees yields between
1.05--1.32x space improvement.
}

We evaluate five graph representations including using \pam{}, Aspen,
\ourtree{} with or without difference encoding, and GBBS.
Aspen uses \ctree{s} as edge trees and leaves vertex trees
uncompressed using \ptree{s}.
GBBS is a state-of-the-art static graph processing library which
represents graphs as static arrays using difference encoding, which
serves as our baseline of graph representation.
\cref{fig:relative_graph_sizes} shows the relative size of each graph
format. We see that the smallest format in all cases is
\defn{\ourtree{} (Diff)}, which applies \ourtree{s} with difference
encoding for both vertex and edge trees.
Using this format yields a space improvement of between 4--9.7x over
just using \ptree{s}.
For the graphs with high average-degree, most of the savings come from
using \ourtree{s} for the edge trees.
Adding difference encoding to both trees yields between 1.05--1.32x
space improvement.
\ourtree{s} are also 1.3--2.6x more space-efficient than Aspen.
Note that \ctree{s} in Aspen are also difference encoded, so the main
difference between the two representations is that \ourgraph{} also
uses \ourtree{s} to chunk the vertex tree, and that \ourtree{s} employ
a deterministic strategy for chunking.
\ourtree{s} with difference encoding achieves consistently lower space
compared with Aspen, ranging between 1.3x for Friendster, our largest
graph, to a maximum space improvement for 2.62x on USA-Road, our
sparsest graph. The space savings come from chunking the vertex trees,
which is not possible in Aspen, since the \ctree{} implementation is
specialized for edge trees.


\begin{table}[!t]
\footnotesize
\setlength{\tabcolsep}{2pt}
\centering

\begin{tabular}[t]{ll ccccccc}
  \toprule
&  & \multicolumn{2}{c}{\bf Aspen} & \multicolumn{4}{c}{\bf Ours} & \\
\cmidrule(lr){3-4}
\cmidrule(lr){5-8}
 &  {Graph} & {FS} & {FS Time}  & {No-FS} & {FS} & {$\frac{\text{FS}}{\text{No-FS}}$} & {FS Time} & {$\frac{\text{Aspen}}{\text{Ours}}$}\\

\midrule
\parbox[t]{2.5mm}{\multirow{3}{*}{\rotatebox[origin=c]{90}{\bf BFS}}}
 &  {\em LiveJournal}         & 21.7 & 3.82    & 19.8 & 17.5 & 1.13x & 1.38    & 1.24x \\
 &  {\em com-Orkut}           & 15.3 & 2.35    & 14.5 & 12.4 & 1.16x & 1.12    & 1.23x \\
 &  {\em Twitter}             & 138  & 37.8    & 125  & 112  & 1.11x & 12.5    & 1.23x \\

\midrule
\parbox[t]{2.5mm}{\multirow{3}{*}{\rotatebox[origin=c]{90}{\bf MIS}}}
& {\em LiveJournal}        & 55.3   & 3.82    & 72.0  & 45.7 & 1.57x  & 1.38  & 1.21x \\
& {\em com-Orkut}          & 70.2   & 2.35    & 96.9  & 69.2 & 1.40x  & 1.12  & 1.01x \\
& {\em Twitter}            & 1022   & 37.8    & 1190  & 971  & 1.22x  & 12.5  & 1.05x \\

\midrule
\parbox[t]{2.5mm}{\multirow{3}{*}{\rotatebox[origin=c]{90}{\bf BC}}}
& {\em LiveJournal}        & 74.6 & 3.82         & 82.1  & 72.3 & 1.13x  & 1.38  & 1.03x \\
& {\em com-Orkut}          & 76.3 & 2.35         & 88.6  & 78.2 & 1.13x  & 1.12  & 0.975x \\
& {\em Twitter}            & 1150 & 37.8         & 2735  & 1030 & 2.65x  & 12.5  & 1.11x \\

\end{tabular}
\caption{\small
\textbf{Parallel running times (in milliseconds) for \aspen{} and
our implementation. }
We show the algorithm performance
without flat snapshots (\textbf{No-FS}), with flat snapshots
(\textbf{FS}), and the time to computing the flat snapshot (\textbf{FS
Time}).}
\label{table:flat-snapshot-comparison}
\vspace{-1.5em}
\end{table}


\myparagraph{Graph Algorithm Performance.}
We study the performance of three fundamental graph kernels:
breadth-first search (BFS), single-source betweenness centrality (BC),
and maximal independent set (MIS). Our implementations are based on
those in Aspen. We study performance using our most space-efficient
version (\ourgraph{}). Following Aspen, our implementation also
supports the \emph{flat snapshot} object, which is an array storing
all vertices in the current graph.  The idea is that instead of
accessing edges for a vertex through the vertex tree (performing tree
traversal), algorithms directly access edge trees through the flat
snapshot.

\cref{table:flat-snapshot-comparison} shows performance results for
three of our graph datasets.
Across all three kernels our implementations are 1.12x faster than
Aspen's implementations on average. We observe that flat snapshots can
be generated 2.09--3.02x faster in \cpam{} due to \ourtree{s} requiring
fewer cache-misses to traverse than \ptree{s} when creating flat
snapshot array. We note that the implementation of edgeMap and other
primitives from Ligra (including constants and other tuning
parameters) are exactly the same in both \cpam{} and Aspen.  Aspen
also difference encodes in its edge trees (represented using
\ctree{s}). The performance improvements that we observe are therefore
a result of \ourtree{s} providing faster flat snapshots, and having
better balance in chunk sizes compared to the randomized approach used
in \ctree{s}.

\hide{
\begin{figure}[!t]
  \includegraphics[width=0.4\textwidth]{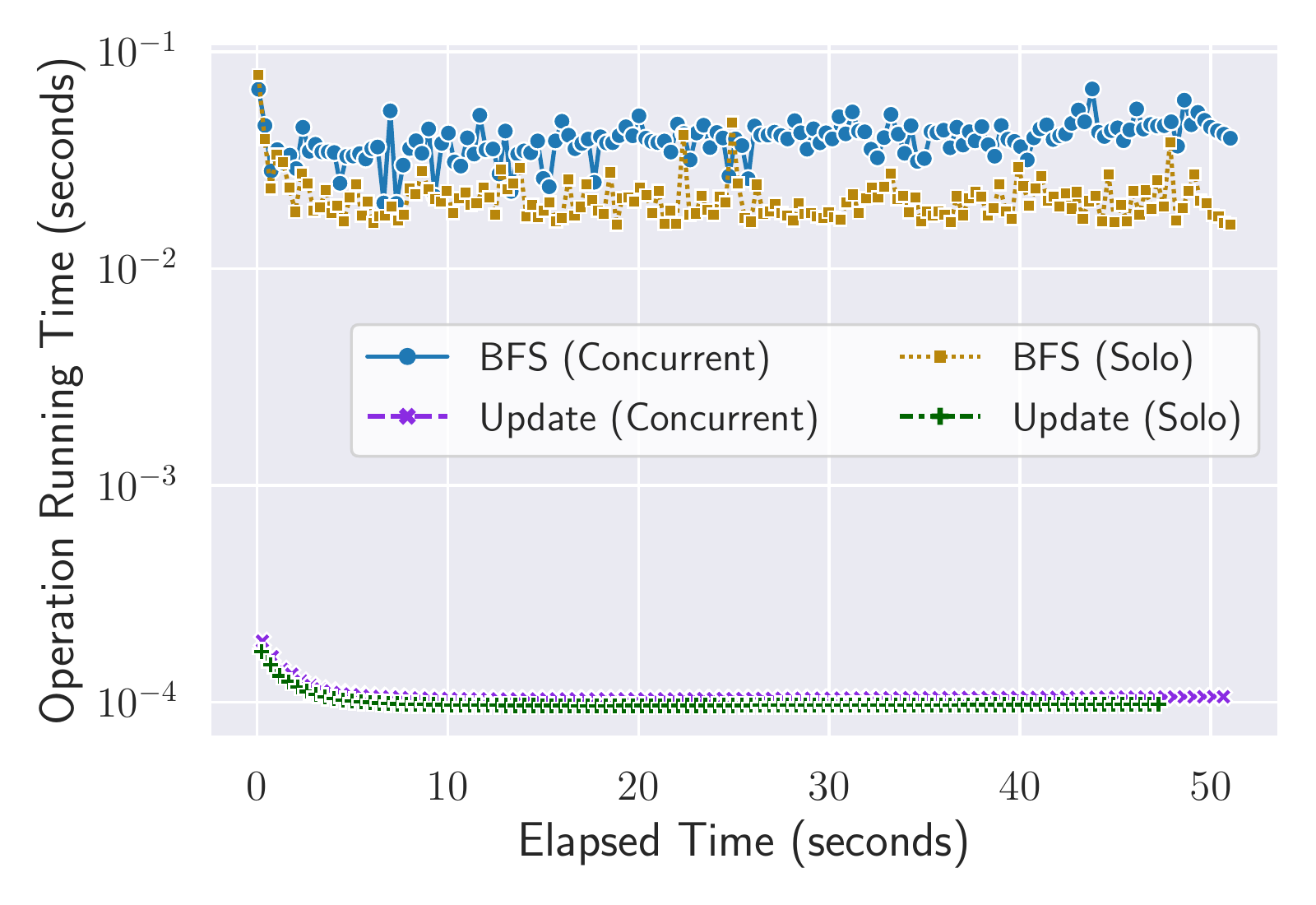}
    \caption{\small Time series plot illustrating operation running
    times when running BFS queries with batch-insertions
    of edges concurrently (Concurrent), and when queries and updates
    are run individually (Solo) on the LiveJournal graph.}\label{fig:concurrent_update_query}
  \vspace{-1em}
\end{figure}
}

\ifx\confversion\undefined
\begin{figure}[!t]
  \includegraphics[width=0.4\textwidth]{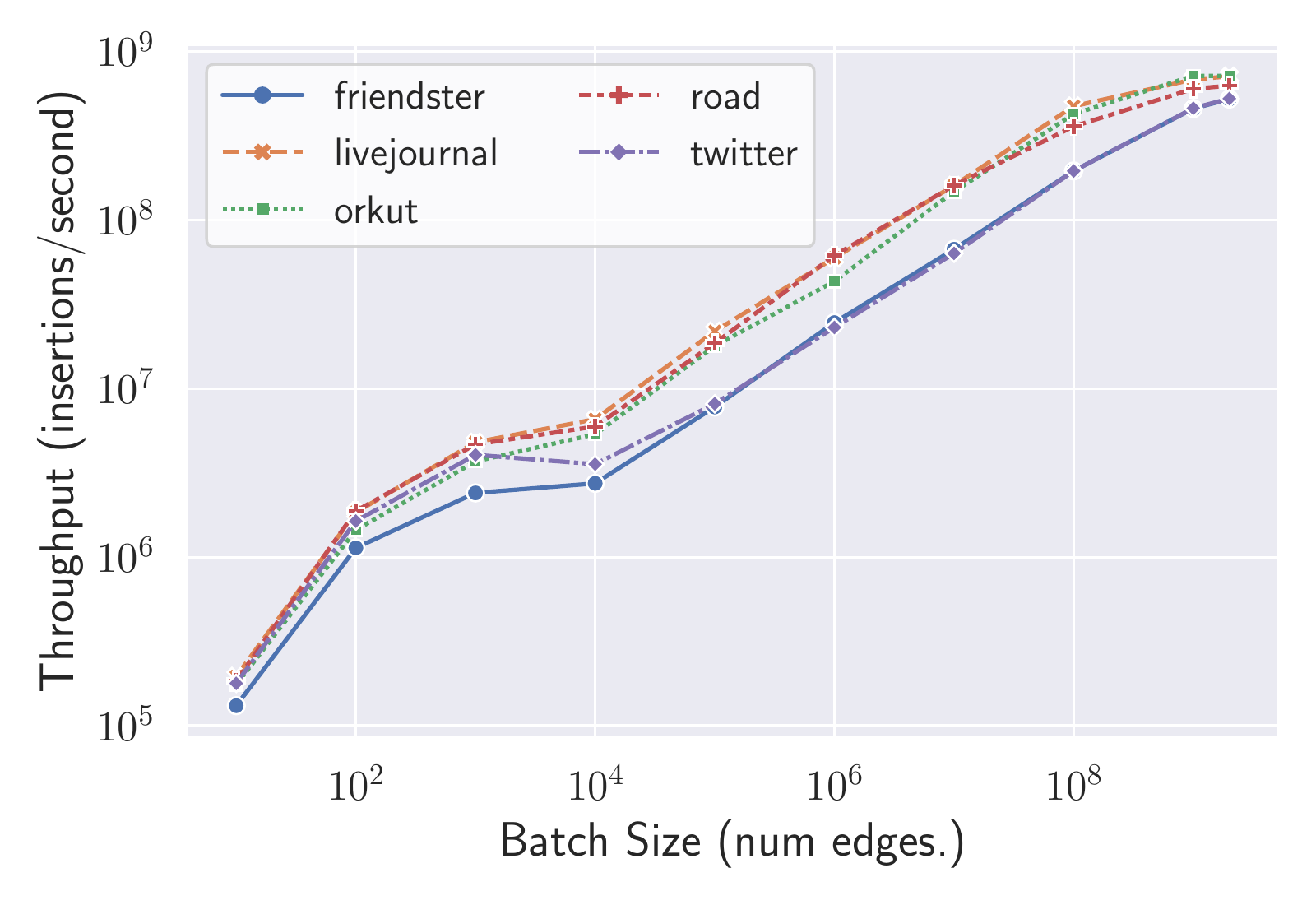}
    \caption{\small \revised{Edge insertion throughput (insertions per second)
    for our graph update algorithms as a function of the number of
    edge insertions in each batch. The throughput for batch deletions
    is similar to that of
    insertions.}}\label{fig:throughput_vs_batchsize}
  \vspace{-1em}
\end{figure}

\myparagraph{Graph Batch-Update Throughput}
In this section we study the performance of our graph representation
using \ourtree{s} when performing updates in varying size batches. We
focus on algorithms for inserting and deleting batches of edges, since
inserting vertices (along with incident edges) can easily be done
using the edge insert primitive, and vertex deletions simply use
\multidelete{} on \ourtree{s}. Our experiment follows the methodology
used in Aspen. To generate updates, we sample directed edges from an
rMAT generator~\cite{ChakrabartiZhanFaloutsos2004} with $a=0.5, b=c=0.1, d=0.3$.
We use the same generator for our concurrent update and query
experiment in Section~\ref{sec:graphs}.
For a batch of size $K$, we generate $K$ directed edge updates from
the stream (note that there can be duplicates) and repeatedly insert
the edges and delete the edges from the batch, reporting the median of
three trials. We note that the times we report include the time to
sort and remove duplicates from the batch.

\cref{fig:throughput_vs_batchsize} shows the throughput of
batch edge insertions (insertions per second) as a function of the
batch size.  We note that the throughput for deletions are close to
that of insertions (within 10\% across all graphs). To remove
clutter, we show results on the five largest graphs in our datasets.
We observe that the throughput of our graph representation improves
with increasing batch size; for the largest batch size, the algorithm
achieves a maximum throughput of between 719M edge insertions per
second for the com-Orkut graph, and a minimum of 527M edge insertions
per second for the Twitter graph. We compared these results with those
of Aspen on the same machine and find that we obtain 1.62x higher
throughput across the three graphs both systems consider in this
experiment, and an average throughput increase of 1.65x across these
graphs.
\fi

\myparagraph{Concurrent Updates and Queries.}
Our last experiment concurrent updates and queries on graphs. The
experiment performs $n$ undirected edge insertions drawn
from the rMAT generator 
\ifx\confversion\undefined
described above. 
\else
(details provided in the full version). 
\fi
We use a batch size of $5$ in the
updates ($10$ directed edges are inserted per batch). We then spawn
two parallel jobs, one performing the updates one batch after the
other, and the other performing BFS queries, one after the other. Both
the updates and queries are parallel (i.e., they internally make use
of parallelism).

\cref{fig:concurrent_update_query} shows the result of the experiment.
We find that the concurrent queries are 1.85x slower on average than
the queries in isolation, and that the concurrent updates are
1.07x slower on average than updates in isolation. In the concurrent
setting, the average latency to make one of the update batches visible
is 100 microseconds, and the updates achieve a throughput of 94,000
undirected edge updates per second.
%
We leave further optimizations and a more in depth study of the graph
setting for future work with our system.
%

\section{Conclusion}

We have presented \ourtree{}, a deterministic compressed ordered map
data structure and an implementation of the structure in a library
\cpam{}.  The important features of \ourtree{}s and its
implementation in \cpam{} include the following.
\setlength{\itemsep}{0pt}
\begin{itemize}[topsep=1.5pt, partopsep=0pt,leftmargin=*]
  \item It is purely functional allowing for persistent snapshots while
    updates are being made, and safe for parallelism.
  \item It supports sequences, ordered sets, ordered maps, and
    augmented maps, with a wide variety of functions on them.
  \item It provides theoretical bounds on work, span, and space.
  \item It achieves fast sequential time and gets up to 100x speedup on 72 cores
    with 144 hyperthreads.
  \item It achieves memory usage that is close to a compressed array and up to an order
    of magnitude smaller than \pam.
  \item It is internally memory manged using reference counting.
  \item It is backward compatible with \pam{}.
  \item It has been used to implement the full functionality of Aspen while
    improving runtime and/or space.
\end{itemize}
\revised{For future work, we are interested in extending \ourtree{s} to support
higher-fanout internal nodes, similar to $B$-trees, which would allow
users to improve query latency at the expense of increased work when
performing updates.} Other future work includes applying
\ourtree{s} to improve space utilization in databases, and to improve
the performance of collection-based applications using non-volatile
memory.

\section*{Acknowledgement}
This work was supported by the National Science Foundation
grants CCF-1901381, CCF-1910030, CCF-1919223, CCF-2103483, and CCF-2119352.

\bibliographystyle{ACM-Reference-Format}
\bibliography{../bibliography/strings,../bibliography/main}

%

\end{document}